\documentclass{article}
\usepackage{amsmath,amsfonts,graphicx,amsthm}
\usepackage[dvipdfm]{hyperref}
\hypersetup{citecolor=blue,colorlinks=true}
\usepackage[authoryear,comma]{natbib}

\theoremstyle{plain}
\newtheorem{theorem}{Theorem}

\theoremstyle{definition}
\newtheorem{definition}{Definition}
\newtheorem{example}{Example}

\theoremstyle{remark}
\newtheorem*{remark}{Remark}

\newcommand{\btheta}{{\boldsymbol{\theta}}}
\newcommand{\wtheta}{{\boldsymbol{\widehat\btheta_\beta}}}

\newcommand{\be}{\begin{equation}}
\newcommand{\ee}{\end{equation}}
\newcommand{\ba}{\begin{eqnarray}}
\newcommand{\ea}{\end{eqnarray}}
\newcommand{\bee}{\begin{equation*}}
\newcommand{\eee}{\end{equation*}}
\newcommand{\baa}{\begin{eqnarray*}}
\newcommand{\eaa}{\end{eqnarray*}}

\begin{document}

\title{\textbf{Generalized Wald-type Tests based on  Minimum Density Power Divergence Estimators}}
\author{Basu, A.$^{1}$; Mandal, A.$^{1}$; Martin, N.$^{2}$ and Pardo, L.$%
^{3} $ \\
$^{1}${\small Indian Statistical Institute, Kolkata 700108, India}\\
$^{2}${\small Department of Statistics, Carlos III University of Madrid,
28038 Getafe (Madrid), Spain}\\
$^{3}${\small Department of Statistics and O.R. Complutense University of
Madrid, 28040 Madrid, Spain} }

\maketitle

\begin{abstract}
In testing of hypothesis the robustness of the tests is an important
 concern. Generally, the maximum likelihood based tests are most efficient under standard regularity conditions, but they
are highly non-robust even under small deviations from the assumed conditions. In this paper we have proposed generalized Wald-type tests based on  minimum density power divergence
estimators for parametric hypotheses. This method avoids the use of nonparametric density estimation and the bandwidth selection. The trade-off between efficiency and robustness is controlled by a tuning parameter $\beta$. The asymptotic distributions
of the test statistics are chi-square with appropriate degrees of freedom. 
The performance of the proposed tests are explored through simulations and real data analysis.
\end{abstract}



\section{Introduction}

\cite{MR1665873} have introduced the minimum density power divergence estimator (MDPDE) that minimizes the density power divergence measure. The 
robustness properties of these estimators have been 
studied in detail by several authors.  However, 
the problem of hypothesis testing based on the density power divergence
measures has only recently been explored (\citealp{MR3011625,basu2013}). The results indicate that the test statistics based on 
the density power divergence measures, when the parameters are estimated by the MDPDEs, have substantially superior performance compared to the 
likelihood ratio test in the presence of outliers. On the other hand, in pure data  these
tests are often competitive to the likelihood ratio tests. So the tests based on the density power divergence  
are very useful practical tools in robust statistics. 

However, sometimes it is not easy to get the expression of the density power
divergence measure between the population densities estimated under the null
hypothesis and under the unrestricted parameter space. Another  
difficulty is that in many situations the asymptotic distributions of the above test
statistics are linear combinations of independent chi-square variables. While none of these
problems are insurmountable, they lead to some reduction in the appeal of 
these otherwise useful tests. 

To overcome these problems we propose a class of generalized Wald-type test
statistics based on   minimum density
power divergence estimators of the model parameters. Our aim is to manipulate
the statistics in a manner that allows us to exploit the nice properties of the 
MDPDEs in constructing the test statistics, and yet come up with asymptotic distributions
which are pure chi-squares rather than linear combinations of independent chi-squares. 
On the whole, we expect to make the tests theoretically and operationally far more simple, 
without compromising the efficiency and the robustness properties of the tests. 

In Section \ref{sec:notation} we have introduced the necessary notations and 
described the existing results. We have proposed 
 Wald-type test statistics in the context of the simple null 
hypothesis in Section \ref{sec:stat}. The problem of the composite null hypothesis is considered in
Section \ref{sec:composite}. Some simulation studies are reported in Section \ref{sec:simulation} to explore the 
behavior of the proposed families of test statistics. Section \ref{sec8} provides a brief  description on the choice of the tuning parameter $\beta$ associated with the proposed Wald-type test. Section \ref{SEC:concluding} has some concluding remarks, and
the proofs are given in the Appendix. 

Throughout this paper, we will refer to the usual Wald test statistic based on the maximum likelihood estimator as the ``classical Wald statistic", whereas the new robust statistics are referred to as ``proposed Wald-type test statistics". 

\section{Notational Set up and Existing Concepts} \label{sec:notation}

Let $\mathcal{G}$ denote the set of all distributions having densities with
respect to a dominating measure (generally the Lebesgue measure or the counting measure). Given any two densities $g$ and $f$ in $%
\mathcal{G}$, the density power divergence between them is defined, as the
function of a nonnegative tuning parameter $\beta $, as 
\begin{equation}
d_\beta(g,f)=\left\{ 
\begin{array}{ll}
\int \left\{ f^{1+\beta }(x)-\left( 1+\frac{1}{\beta }\right) f^{\beta
}(x)g(x)+\frac{1}{\beta }g^{1+\beta }(x)\right\} dx, & \text{for}\mathrm{~}%
\beta >0, \\[2ex] 
\int g(x)\log \left( \frac{g(x)}{f(x)}\right) dx, & \text{for}\mathrm{~}%
\beta =0.%
\end{array}%
\right.  \label{2.1}
\end{equation}%
The case corresponding to $\beta =0$ may be derived from the general case by
taking the continuous limit as $\beta \rightarrow 0$, and in this case $%
d_0(g,f)$ is the classical Kullback-Leibler divergence. The quantities
defined in equation (\ref{2.1}) are genuine divergences in the sense $%
d_\beta(g,f)\geq 0$ for all $g,f\in \mathcal{G}$ and all $\beta \geq 0$,
and $d_\beta(g,f)$ is equal to zero if and only if the densities $g$ and $%
f$ are identically equal. More details about inference based on divergence measures can be found in \cite{MR2183173} and \cite{MR2830561}.

We consider a parametric model of densities $\{f_{\btheta}:{%
\btheta}\in \Theta \subset {\mathbb{R}}^{p}\}$ and suppose that
we are interested in the estimation of ${\btheta}$. Let $G$
represent the distribution function corresponding to the density $g$. The
minimum density power divergence functional $T_\beta(G)$ at $G$ is
defined by the requirement $d_\beta(g,f_{T_\beta(G)})=\min_{{{%
\btheta}}\in \Theta }d_\beta(g,f_{\btheta})$.
Clearly the term $\int g^{1+\beta }(x)dx$ in (\ref{2.1}) has no role in the
minimization of $d_\beta(g,f_{\btheta})$ over ${{%
\btheta}}\in \Theta $. Thus the essential objective function to
be minimized in the computation of the minimum density power divergence
functional $T_\beta(G)$ reduces to 
\begin{equation}
\int f_{%
\btheta}^{1+\beta }(x)dx-\left( 1+\frac{1}{\beta }\right)
\int f_{\btheta}^{\beta }(x)g(x)dx = \int f_{%
\btheta}^{1+\beta }(x)dx-\left( 1+\frac{1}{\beta }\right)
\int f_{\btheta}^{\beta }(x)dG(x).  \label{2.21}
\end{equation}
Notice that in the above objective function the density $g$ appears only as
a linear term (unlike, say, the objective function of  the minimum Hellinger
distance functional where the square root of the density $g$ is the relevant
quantity). Given a random sample $X_{1},\ldots ,X_{n}$ from the
distribution $G$, we can therefore approximate the objective function in (\ref{2.21}) by
replacing $G$ with its empirical estimate $G_{n}$. For a given tuning
parameter $\beta $, the MDPDE $\widehat\btheta%
_\beta$ of $\btheta$ can be obtained by minimizing 
\begin{equation}
 \begin{split}
\int f_{\btheta}^{1+\beta }(x)dx & - \left( 1+\frac{1}{\beta }%
\right) \int f_{\btheta}^{\beta }(x)dG_{n}(x)\\
 &=\int f_{{%
\btheta}}^{1+\beta }(x)dx-\left( 1+\frac{1}{\beta }\right) 
\frac{1}{n}\sum_{i=1}^{n}f_{\btheta}^{\beta }(X_{i})
 =\frac{1}{n}\sum_{i=1}^{n}V_\btheta(X_{i})  \label{2.22}
\end{split}%
\end{equation}
over ${\btheta}\in \Theta $, where $V_{\btheta%
}(x)=\int f_{\btheta}^{1+\beta }(y)dy-\left( 1+\frac{1}{\beta 
}\right) f_{\btheta}^{\beta }(x)$. In the special case $\beta
=0$, we must minimize the expression $-\frac{1}{n}\sum_{i=1}^{n}\log f_{%
\btheta}(X_{i})$; the corresponding minimizer turns out to be
the maximum likelihood estimator (MLE) of $\btheta$. The
minimization of the expression in (\ref{2.22}) over ${\btheta}
$ does not require the use of a nonparametric density estimate of the true
unknown distribution $G$. Existing theory (e.g. \citealp{de1992smoothing})
shows that in general there is hardly any advantage in introducing
smoothing for such functionals which may be empirically estimated using the
empirical distribution function alone, except in very special cases. It is therefore natural to substitute $G$ with $G_n$ in such situations. 

Let $\boldsymbol{u}_{\btheta}(x)=\frac{\partial }{\partial 
\btheta}\log f_{\btheta}(x)$ be the score
function of the model. Under differentiability of the model the minimization
of the objective function in equation (\ref{2.22}) leads to an estimating
equation of the form 
\begin{equation}
\frac{1}{n}\sum_{i=1}^{n}\boldsymbol{u}_{\btheta}(X_{i})f_{{%
\btheta}}^{\beta }(X_{i})-\int \boldsymbol{u}_{{\boldsymbol{%
\theta }}}(x)f_{\btheta}^{1+\beta }(x)dx=\boldsymbol{0}_p,
\label{2.2}
\end{equation}%
which is an unbiased estimating equation under the model. Here $\boldsymbol{0}_p$
denotes the null vector of dimension $p$.  Since the
corresponding estimating equation weights the score $\boldsymbol{u}_{%
\btheta}(X_{i})$ with the power of the density $f_{\boldsymbol{%
\theta }}^{\beta }(X_{i})$, the outlier resistant behavior of the estimator
is intuitively apparent. See \cite{MR1665873} and \cite{MR1859416} for
more details.

The functional $T_\beta(\cdot)$ is Fisher consistent; it takes the value $%
\btheta_0$ when the true density $g=f_{\btheta%
_0}$ is in the model. When it is not, $\btheta_{\beta
}^{g}=T_\beta(G)$ represents the best fitting parameter. For brevity we
will suppress the $g$ superscript in the notation for $\btheta%
_\beta^{g}$; $f_{\btheta_\beta}$ is the model element
closest to the density $g$ in the minimum density power divergence sense
corresponding to the tuning parameter $\beta$.

Let $g$ be the true data generating density, and $\btheta%
_\beta=T_\beta(G)$ be the best fitting parameter. To set up the
notation we define the quantities 
\begin{eqnarray}
 \boldsymbol{J}_\beta(\btheta) &=&\int \boldsymbol{u}_{{%
\btheta}}(x)\boldsymbol{u}_{\btheta}^{T}(x)f_{{%
\btheta}}^{1+\beta }(x)dx \nonumber\\
& +&\int \{\boldsymbol{I}_{\boldsymbol{%
\theta }}(x)-\beta \boldsymbol{u}_\btheta(x)\boldsymbol{u}_{%
\btheta}^{T}(x)\}\{g(x)-f_\btheta(x)\}f_{{%
\btheta}}^{\beta }(x)dx,  \label{2.3} \\
 \boldsymbol{K}_\beta(\btheta) &=&\int \boldsymbol{u}_{{%
\btheta}}(x)\boldsymbol{u}_{\btheta}^{T}(x)f_{{%
\btheta}}^{2\beta }(x)g(x)dx-\boldsymbol{\xi }_\beta({{%
\btheta}})\boldsymbol{\xi }_\beta^{T}({\btheta%
}),  \label{2.4}
\end{eqnarray}
where $\boldsymbol{\xi }_\beta({\btheta})=\int \boldsymbol{%
u}_\btheta(x)f_\btheta^{\beta }(x)g(x)dx$, and $%
\boldsymbol{I}_\btheta(x)=-\frac{\partial }{\partial 
\btheta^T}\boldsymbol{u}_\btheta(x)$ is the so-called Fisher information function at the model. The following results, proved in
\cite{MR2830561}, form the basis of our subsequent developments. We assume
that the conditions D1--D5 of \citeauthor{MR2830561} (2011, p. 304) are true. Then
the following results hold.

\begin{enumerate}
\item The minimum density power divergence estimating equation (\ref{2.2})
has a consistent sequence of roots $\wtheta$, i.e. $\widehat{\boldsymbol{%
\theta }}_\beta\underset{n \rightarrow \infty }{\overset{p}{%
\longrightarrow }}\btheta_\beta$.

\item $n^{1/2}(\wtheta-\btheta%
_\beta)$ has an asymptotic multivariate normal distribution with 
  mean vector zero and covariance matrix $\boldsymbol{J}^{-1}\boldsymbol{K}%
\boldsymbol{J}^{-1}$, where $\boldsymbol{J}=\boldsymbol{J}_\beta({{%
\btheta}}_\beta)$, $\boldsymbol{K}=\boldsymbol{K}_\beta({{%
\btheta}}_\beta)$ are as in (\ref{2.3}) and (\ref{2.4}). 
\end{enumerate}
In Sections \ref{sec:stat} and \ref{sec:composite} we describe the proposed Wald-type statistics for the simple and composite null hypothesis respectively.

\section{The Simple Null Hypothesis}\label{sec:stat}

Let $X_{1}, \ldots, X_{n}$ be a random sample of size $n$ 
from a distribution modeled by the probability density function $f_\btheta(x)$, where 
$\btheta \in \Theta \subset \mathbb{R}^p$ and $x \in \mathcal{X}$, the sample space.
In this section we  define a family of Wald-type
test statistics based on MDPDEs  for testing the null hypothesis%
\begin{equation}
H_0:\btheta = \btheta_0~{\rm against}~H_{1}:\btheta \neq \btheta_0 ,
\label{3.1}
\end{equation}%
which will henceforth be referred to as the proposed Wald-type statistics.

\begin{definition}
Let $\wtheta$ be the MDPDE of $\btheta$. The family of proposed Wald-type test statistics for testing the null hypothesis 
in (\ref{3.1}) is given by 
\begin{equation}
W_n=n\left( \wtheta-\btheta_0\right)^{T}
\left( \boldsymbol{J}_\beta^{-1}( \btheta_0) 
\boldsymbol{K}_\beta( \btheta_0) 
\boldsymbol{J}_\beta^{-1}( \btheta_0) \right)^{-1}
\left( \wtheta-\btheta_0\right) ,
\label{3.2}
\end{equation}%
where the matrices $\boldsymbol{J}_\beta( \btheta_0)$ 
and $\boldsymbol{K}_\beta( \btheta_0)$ are as defined in (\ref{2.3}) 
and (\ref{2.4}) respectively.
\end{definition}

When $\beta = 0$, $\wtheta$ coincides with the maximum likelihood 
estimator $\widehat\btheta$ of $\btheta$, and $\boldsymbol{J}_\beta^{-1}
(\boldsymbol{\theta }_0)\boldsymbol{K}_\beta(\btheta_0)\boldsymbol{J}_\beta^{-1}(\btheta_0)$ coincides
with the inverse of the Fisher information matrix, and thus we recover the classical Wald 
statistic. 
Many interesting applications of the classical Wald test statistic can be seen in the statistical literature.
For example, \cite{lemonte2013nonnull} has presented an application on generalized linear models with dispersion covariates. 

The asymptotic distribution of $W_n$ is presented in the next theorem. The result follows 
easily from the asymptotic distribution of MDPDEs considered in Section \ref{sec:notation},
so we skip the proof. 

\begin{theorem}
The asymptotic null distribution of the proposed Wald-type test statistic given in (\ref%
{3.2}) is a chi-square distribution with $p$ degrees of freedom. 
\end{theorem}

In many cases the power function of this testing procedure cannot be
calculated explicitly. In the following theorem we present a useful
asymptotic result for approximating the power function of the proposed Wald-type test
statistics given in (\ref{3.2}).

\begin{theorem}
Let $\btheta^*$ be the true value of parameter with $%
\btheta^*\neq \btheta_0.$ Then the convergence 
\begin{equation*}
n^{1/2}\left( l\left( \wtheta\right) -l\left( 
\btheta^*\right) \right) \underset{n\rightarrow
\infty }{\overset{\mathcal{L}}{\longrightarrow }}\mathit{N}_p(\boldsymbol{0}_p,\sigma %
_{W}^{2}\left( \btheta^*\right) )
\end{equation*}%
holds, where 
\begin{equation*}
l( \btheta) =\left( \btheta-\boldsymbol{%
\theta }_0\right) ^{T}\left( \boldsymbol{J}_\beta^{-1}\left( \boldsymbol{%
\theta }_0\right) \boldsymbol{K}_\beta\left( \btheta%
_0\right) \boldsymbol{J}_\beta^{-1}( \btheta_0)
\right) ^{-1}\left( \btheta-\btheta_0\right),
\end{equation*}%
and 
\begin{equation}
\sigma_W^{2}\left( \btheta^{\ast
}\right) =4\left( \btheta^*-\btheta%
_0\right) ^{T}\left( \boldsymbol{J}_\beta^{-1}\left( \boldsymbol{\theta 
}^*\right) \boldsymbol{K}_\beta\left( \btheta^{\ast
}\right) \boldsymbol{J}_\beta^{-1}\left( \btheta^{\ast
}\right) \right) ^{-1}\left( \btheta^*-\btheta%
_0\right) .  \label{3.3}
\end{equation}
\label{th:main}
\end{theorem}

\begin{proof}
See the Appendix.
\hspace*{\fill}\bigskip
\end{proof}

\begin{remark}
We  approximate the power function $\beta_{W_n}$ of the proposed Wald-type
test statistics at $\btheta^*$ by%
\begin{align}
 \beta_{W_n}\left( \btheta^*\right) &= P_{\btheta^*} \left(
W_n>\chi_{p,\alpha }^{2}\right) \nonumber\\
&= P_{\btheta^*} \left( l\left( \wtheta\right) -l\left( 
\btheta^*\right) >\frac{\chi_{p,\alpha }^{2}}{n}-l\left( 
\btheta^*\right) \right) \nonumber\\
&= P_{\btheta^*} \left( n^{1/2}\left( l\left( \widehat\btheta_{\beta
}\right) -l\left( \btheta^*\right) \right) >n^{1/2}%
\left( \frac{\chi_{p,\alpha }^{2}}{n}-l\left( \btheta^{\ast
}\right) \right) \right) \nonumber\\
&= P_{\btheta^*} \left( n^{1/2}\frac{\left( l\left( \widehat\btheta_{\beta
}\right) -l\left( \btheta^*\right) \right) }{
\sigma_{W}\left( \btheta^*\right) }>\frac{n^{1/2}}{%
\sigma_{W}\left( \btheta^*\right) }\left( 
\frac{\chi_{p,\alpha }^{2}}{n}-l\left( \btheta^*\right)
\right) \right) \nonumber\\
&=1-\Phi_n\left( \frac{n^{1/2}}{\sigma_{W}\left( 
\btheta^*\right) }\left( \frac{\chi_{p,\alpha }^{2}}{n}%
-l\left( \btheta^*\right) \right) \right),
\label{n}
\end{align}
where $\Phi_{n}(\cdot)$ is a sequence of distribution functions tending uniformly
to the standard normal distribution function $\Phi(\cdot).$ Here 
$\chi^2_{p,\alpha}$ is the $100(1-\alpha)$ percentile point of chi-square distribution with $p$ degrees of freedom, i.e. $P(\chi^2_p>\chi^2_{p,\alpha})=\alpha$.
It is clear that 
\begin{equation*}
\lim_{n\rightarrow \infty }\beta_{W_n} (\btheta^*) =1,
\end{equation*}%
for all $\alpha \in \left( 0,1\right).$ Therefore, the test is consistent in
the sense of Fraser \citep{MR0093863}.
\end{remark}

Based on this result given in (\ref{n}) we can calculate the required sample size in order 
to get a given value for the power.  If we want to get a power of $\beta^*=\beta_{W_n} (\btheta^*)$, we need a sample size $n$ given by
\begin{equation}
 n=[n^*]+1, \ n^* = \frac{A+B + \sqrt{A (A+2B)}}{2 l^2(\btheta^*)},
\end{equation}
where $[x]$ is the largest integer less than or equal to $x$, $A=\sigma^2_W(\btheta^*) \left(\Phi^{-1}(1-\beta^*)\right)^2$ and $B = \frac{1}{2}\chi^2_{p,\alpha} l(\btheta^*)$. 

Now we are going to derive the asymptotic distribution of $W_n$ under
contiguous alternative hypotheses described by 
\begin{equation}
H_{1,n}:\btheta_{n}=\btheta_0+ n^{-1/2}\boldsymbol{d%
},  \label{3.4}
\end{equation}
where ${\boldsymbol{d}}$ is a fixed vector in $\mathbb{R}^p$ such that  $\btheta_{n} \in \Theta  \subset \mathbb{R}^p$ for all $n$.
\begin{theorem}
Under the contiguous alternative hypotheses given in (\ref{3.4})  the
asymptotic distribution of the proposed Wald-type test statistic $W_n$ is a non-central
chi-square with $p$ degrees of freedom and non-centrality
parameter 
\begin{equation}
\delta =\boldsymbol{d}^{T}\boldsymbol{J}_\beta^{-1}\left( \btheta%
_0\right) \boldsymbol{K}_\beta( \btheta_0) 
\boldsymbol{J}_\beta^{-1}( \btheta_0) \boldsymbol{d}%
. \label{d}
\end{equation}
\label{Th:1}
\end{theorem}

\begin{proof}
See the Appendix.
\hspace*{\fill}\bigskip
\end{proof}

\begin{remark}
Using the last theorem we can get an approximation to the power
function  at the contiguous alternative $\btheta_{n}=\btheta_0+ n^{-1/2}
\boldsymbol{d}$ through the relation
\begin{equation*}
\beta_{W_n}\left( \btheta_n\right) =1-G_{\chi
_{p}^{2}(\delta )}\left( \chi_{p,\alpha }^{2}\right) ,
\end{equation*}%
where $G_{\chi_{p}^{2}(\delta )}\left( \chi_{p,\alpha }^{2}\right) $ is
the distribution function of a non-central chi-square, with $p$
degrees of freedom and non-centrality parameter 
$\delta$ given in equation (\ref{d}), 
evaluated at the point $\chi_{p,\alpha }^{2}.$ 

It may be observed that this expression permits us to obtain an approximation of the 
power function at a generic point $\btheta^*$, because we can consider $\boldsymbol{d} = n^{1/2}( \btheta^* - \btheta_0)$
and then $\btheta_n=\btheta^*$ .
\end{remark}

\begin{example} \label{Exp_example}
Let  $X_{1},\ldots ,X_{n}$ be a random sample of size $n$ from $%
\mathcal{E}xp(\theta )$, the exponential distribution with mean $\theta $. Suppose we want to
test the  hypothesis $H_0:\theta =\theta_0$, against $%
H_{1}:\theta \neq \theta_0$. It is 
easy to show that the minimum density power divergence estimator $%
\widehat{\theta }_\beta$ of $\theta $ can be obtained by
iteratively solving the equation%
\begin{equation}
\widehat{\theta }_\beta=\frac{\sum_{i=1}^{n}X_{i}\exp \left\{ -\frac{%
\beta X_{i}}{\widehat{\theta }_\beta}\right\} }{\sum_{i=1}^{n}\exp
\left\{ -\frac{\beta X_{i}}{\widehat{\theta }_\beta}\right\} -\frac{%
n\beta }{(1+\beta )^{2}}}   \label{3.5} .
\end{equation}%
Note that, putting $\beta=0$ in equation (\ref{3.5}), we get an explicit expression for  the MLE as $\widehat{\theta }%
_0=n^{-1}\sum_{i=1}^{n}X_{i}.$ Standard calculation with the matrices $\boldsymbol{J}_\beta( \btheta_0)$ 
and $\boldsymbol{K}_\beta( \btheta_0)$, as defined in (\ref{2.3}) and (\ref{2.4}) respectively, gives us
\begin{equation*}
n^{1/2}\left( \widehat{\theta }_\beta-\theta %
_0\right) \underset{n\rightarrow \infty }{\overset{\mathcal{L}}{\longrightarrow }}%
\mathit{N}(0,h(\beta )\theta_0^{2}),
\end{equation*}%
where 
\begin{equation}
h(\beta )=\frac{(1+\beta )^{2}P(\beta )}{(1+\beta ^{2})^{2}(1+2\beta )^{3}}%
,\mbox{ and }P(\beta )=1+4\beta +9\beta ^{2}+14\beta ^{3}+13\beta ^{4}+8\beta
^{5}+4\beta ^{6}.  \label{3.6}
\end{equation}
Therefore, for testing  hypothesis $H_0:\theta =\theta_0$, against $%
H_{1}:\theta \neq \theta_0$, the proposed Wald-type test statistic is given by
\begin{equation}
W_n = \frac{n(\widehat{\theta }_\beta-\theta_0)^{2}}{h(\beta )\theta_0^{2}%
} ,
\label{3.7}
\end{equation}
whose asymptotic distribution, under the null hypothesis, is a chi-square  with one 
degree of freedom. 
\end{example}
\section{The Composite Null Hypothesis}\label{sec:composite}

We will now consider the problem of testing composite null hypothesis.
In testing of a composite null hypothesis, the restricted parameter
space $\Theta_0\subset \Theta $ is often defined by a set of $r$ restrictions
of the form 
\begin{equation}
\boldsymbol{m}(\boldsymbol{\theta )=0}_{r} , \label{2.5}
\end{equation}%
 where $\boldsymbol{m}:\mathbb{R}^{p}\rightarrow \mathbb{R}^{r}$ (see \citealp{MR595165}). Assume that the $p\times r$ matrix 
\begin{equation}
\boldsymbol{M}( \btheta) =\frac{\partial \boldsymbol{m}%
^{T}(\boldsymbol{\theta )}}{\partial \btheta}  \label{2.6}
\end{equation}%
exists and is continuous in $\btheta$, and rank$\left( \boldsymbol{M}%
( \btheta) \right) =r$, where $r \leq p$. 
 Let $X_{1}, \ldots, X_{n}$ be a random sample of size $n$ 
from a distribution modeled by probability density function $f_\btheta(x)$, where 
$\btheta \in \Theta \subset \mathbb{R}^p$ and $x \in \mathcal{X}$, the sample space. Our interest is in testing 
the hypothesis 
 \begin{equation}
H_0:\btheta\in \Theta_0~\text{against}~H_{1}:\boldsymbol{%
\theta }\notin \Theta_0,  \label{4.1}
\end{equation}%
where $\Theta_0$ is a subset of the parameter space $\Theta$.

\begin{definition}
Let $\wtheta$ be the MDPDE of $\btheta$. 
The family of proposed Wald-type test statistics for testing the null hypothesis in 
(\ref{4.1}) is given by 
\begin{equation}
W_{n}=n\boldsymbol{m} ^{T}\left( \wtheta\right) \left[
\boldsymbol{M}^{T}\left( \wtheta\right)\boldsymbol{J}_\beta^{-1}
\left( \wtheta\right) \boldsymbol{K}_\beta\left( \wtheta\right)%
\boldsymbol{J}_\beta^{-1}\left( \wtheta\right)\boldsymbol{M}
\left( \wtheta\right)\right] ^{-1}\boldsymbol{m}\left( \wtheta\right) ,  \label{4.2}
\end{equation}%
where the matrices $\boldsymbol{J}%
_\beta, \boldsymbol{K}_{\beta}$ and $\boldsymbol{M}$ are defined in equations  (\ref%
{2.3}), (\ref{2.4}) and (\ref{2.6}) respectively, and the function $\boldsymbol{m}$ is defined
in (\ref{2.5}).
\end{definition}

As in the simple null hypotheses case, when $\beta=0$, the Wald-type test statistic reduces to the classical Wald test.
In the next theorem we present the asymptotic distribution of $W_{n}.$

\begin{theorem} \label{theorem_composite}
The asymptotic null distribution of the proposed Wald-type test statistics given in (\ref%
{4.2}) is  chi-square  with $r$ degrees of freedom.
\end{theorem}

\begin{proof}
See the Appendix.
\hspace*{\fill}\bigskip
\end{proof}

It should be noted that the validity of the results presented are based on a crucial assumption that the unknown parameter vector $\btheta$ is an interior point in $\Theta$. A common violation occurs when the null hypothesis constrains some of the parameters to lie on the boundary of the parameter space. Asymptotic theory under nonstandard conditions, in the case of MLE and likelihood ratio test, has been developed in \cite{chernoff1954distribution}, \cite{chant1974asymptotic}, \cite{self1987asymptotic}, \cite{susko2013likelihood} and references therein.  One could, along similar lines, develop a theory for testing null hypothesis using proposed Wald-type statistics involving MDPDEs that cover parameters in the boundary of the parameter space. This will require a detailed investigation appropriate for a separate paper, and has not been carried out in the present work. For the rest of the paper we assume that the unknown parameter lies in the interior of the parameter space.

Consider the null hypothesis $H_0:\btheta\in \Theta_0\subset
\Theta $. By Theorem \ref{theorem_composite}, the null hypothesis is rejected if 
$W_{n}\geq \chi_{r,\alpha }^{2}$. The following theorem can be used to
approximate the power function. Assume that $\btheta^*\notin
\Theta_0$ is the true value of the parameter so that the unrestricted estimator $\widehat{\boldsymbol{%
\theta }}_\beta\overset{p}{\underset{n\rightarrow \infty }{%
\longrightarrow }}\btheta^*$. 

\begin{theorem} Suppose
\begin{equation*}
l^*(\btheta_1, \btheta_2) =\boldsymbol{m} ^{T}\left( 
\btheta_1 \right) \left[ \boldsymbol{M}^{T}(\btheta_2)%
\boldsymbol{J}_\beta^{-1}(\btheta_2)\boldsymbol{K}_{\beta
}(\btheta_2)\boldsymbol{J}_\beta^{-1}(\btheta_2)
\boldsymbol{M}(\btheta_2)\right] ^{-1}\boldsymbol{m}\left( 
\btheta_1\right).
\end{equation*}
Then
\begin{equation*}
n^{1/2}\left( l^*\left(\wtheta, \wtheta\right)
-l^*\left(\btheta^*, \btheta^*\right) \right) \underset{%
n\rightarrow \infty }{\overset{\mathcal{L}}{\longrightarrow }}\mathit{N}(0%
\boldsymbol{,\sigma }_{W}^{2}\left( \btheta^*\right) ),
\end{equation*}%
where
\begin{equation}
\sigma_{W}^{2}\left(  \btheta^*\right)
=\left( \frac{\partial l^*(\btheta, \btheta^*) }{%
\partial \btheta}\right)_{\btheta =\btheta^*
}^{T}\boldsymbol{J}_\beta^{-1}(\btheta^*)\boldsymbol{K}%
_\beta(\btheta^*)\boldsymbol{J}_\beta^{-1}(\btheta^*)
\left( \frac{\partial l^*(\btheta, \btheta^*) }{%
\partial \btheta}\right)_{\btheta =\btheta^*
}^{T}.  \label{4.3}
\end{equation}
\label{Th:10}
\end{theorem}

\begin{proof}
See the Appendix.
\hspace*{\fill}\bigskip
\end{proof}

We may also find an approximation of the power of $W_{n}$ at an alternative
close to the null hypothesis. Let $\btheta_{n}\in \Theta -\Theta
_0$ be a given alternative, and let $\btheta_0$ be the element
in $\Theta_0$ closest to $\btheta_{n}$ in terms of the Euclidean distance.
One possibility to introduce contiguous alternative hypotheses in this context is
to consider a fixed $\boldsymbol{d}\in \mathbb{R}^p$ and to permit $%
\btheta_{n}$ move towards $\btheta_0$ as $n$ increases
through the relation 
\begin{equation}
H_{1,n}:\btheta_{n}=\btheta_0+n^{-1/2}\boldsymbol{d}.
\label{4.3.1}
\end{equation}%
A second approach is to relax the condition $\boldsymbol{m}\left( \boldsymbol{%
\theta }\right) =\boldsymbol{0}_r$ defining $\Theta_0.$ Let $\boldsymbol{\delta} \in 
\mathbb{R}^r$ and consider the following sequence  of parameters $\{\btheta
_{n}\}$ moving towards $\btheta_0$ according to the set up
\begin{equation}
H_{1,n}^*:\boldsymbol{m}\left( \btheta_{n}\right)
=n^{-1/2}\boldsymbol{\delta} .  \label{4.3.2}
\end{equation}%
Note that a Taylor series expansion of $\boldsymbol{m}\left( \btheta%
_{n}\right) $ around $\btheta_0$ yields 
\begin{equation}
\boldsymbol{m}\left( \btheta_{n}\right) =\boldsymbol{m}\left( 
\btheta_0\right) +\boldsymbol{M} ^{T}\left( \btheta%
_0\right) \left( \btheta_{n}-\btheta_0\right)
+o\left( \left\Vert \btheta_{n}-\btheta_0\right\Vert
\right) .  \label{4.4}
\end{equation}%
By substituting $\btheta_{n}=\btheta_0+n^{-1/2}%
\boldsymbol{d}$ in (\ref{4.4}) and taking into account that\textbf{\ }$%
\boldsymbol{m}( \btheta_0) =\boldsymbol{0}_r,$ we get 
\begin{equation}
\boldsymbol{m}\left( \btheta_{n}\right) =n^{-1/2}\boldsymbol{M} ^{T}%
( \btheta_0) \boldsymbol{d}+o\left( \left\Vert 
\btheta_{n}-\btheta_0\right\Vert \right) . \label{mt}
\end{equation}%
So the equivalence relationship between the hypotheses $H_{1,n}$ and $H_{1,n}^*$ is 
\begin{equation}
\boldsymbol{\delta }=%
\boldsymbol{M} ^{T}( \btheta_0) \boldsymbol{d} \mbox{ as } n \rightarrow \infty. 
\label{eqiv}
\end{equation}

In the following theorem we show the asymptotic distributions of the test statistics 
under the alternative hypotheses $H_{1,n}$ and $H_{1,n}^*$  as given in (\ref{4.3.1}) and (\ref{4.3.2}) respectively.
\begin{theorem} The asymptotic distribution of $W_n$ is given by
\begin{enumerate}
\item[i)] $W_{n}\underset{n \rightarrow \infty }{\overset{\mathcal{L}}{%
\longrightarrow }}\chi_{r}^{2}\left( \boldsymbol{d}^{T}\boldsymbol{M}\left( 
\btheta_0\right) \left[ \boldsymbol{M}^{T}(\btheta%
_0)\boldsymbol{J}_\beta^{-1}(\btheta_0)\boldsymbol{K}%
_\beta(\btheta_0)\boldsymbol{J}_\beta^{-1}(\boldsymbol{%
\theta }_0)\boldsymbol{M}(\btheta_0)\right] ^{-1}%
\boldsymbol{M} ^{T}( \btheta_0) \boldsymbol{d}\right) $
under $H_{1,n}$ given in (\ref{4.3.1}).

\item[ii)] $W_{n}\underset{n \rightarrow \infty }{\overset{\mathcal{L}}{%
\longrightarrow }}\chi_{r}^{2}\left( \boldsymbol{\delta }^{T}\left[ 
\boldsymbol{M}^{T}(\btheta_0)\boldsymbol{J}_\beta^{-1}(%
\btheta_0)\boldsymbol{K}_\beta(\btheta_0)%
\boldsymbol{J}_\beta^{-1}(\btheta_0)\boldsymbol{M}(%
\btheta_0)\right] ^{-1}\boldsymbol{\delta }\right) $ under $%
H_{1,n}^*$ given in (\ref{4.3.2}).
\label{Th:11}
\end{enumerate}
\end{theorem}

\begin{proof}
See the Appendix.
\hspace*{\fill}\bigskip
\end{proof}

While our theory is perfectly general, we discuss two special cases in Sections \ref{example_normal} and \ref{sec:weibull} -- those of the Normal and Weibull models -- where immediate applications our theory can be envisaged.

\subsection{The Normal Distribution}
\label{example_normal}
Under the $N(\mu ,\sigma ^{2})$ model, consider the problem of testing 
\begin{equation}
H_0:\mu =\mu_0\text{ against }H_{1}:\mu \neq \mu_0,  \label{A}
\end{equation}%
where $\sigma ^{2}$ is an unknown nuisance parameter. In this case the
 parameter space is given by $\Theta =\{(\mu ,\sigma
)\in {\mathbb{R}}^{2}|\mu \in {\mathbb{R}},\sigma \in {\mathbb{R}}^{+}\}$,
and the parameter space under the null distribution is $\Theta_0=\{(\mu ,\sigma )\in {\mathbb{R}}^{2}|\mu =\mu_0,\sigma
\in {\mathbb{R}}^{+}\}$. If we consider the function $m(%
\btheta)=\mu -\mu_0,$ where $\btheta=\left( \mu
,\sigma \right) ^{T}$, the null hypothesis $H_0$ can be written as 
\begin{equation*}
H_0:m(\btheta)=0.
\end{equation*}%
 We  observe that in
our case $\boldsymbol{M}( \btheta) =\left( 1,0\right)
^{T}$. Based on (\ref{2.2}) and taking into account the fact that $f_{%
\btheta}(x)$ is the normal density with mean $\mu $ and
variance $\sigma ^{2}$, the estimator $\widehat\btheta_{\beta
}=(\widehat{\mu }_\beta,\widehat{\sigma }_\beta)^{T}$ of $\boldsymbol{%
\theta }$ is the solution of the system of nonlinear
equations 
\begin{equation*}
\left\{ 
\begin{array}{c}
\frac{\partial }{\partial \mu }\frac{1}{\sigma ^{\beta }\left( 2\pi \right)
^{\frac{\beta }{2}}}\left( \frac{1}{n\beta }\sum_{i=1}^{n}\exp \left\{ -%
\frac{\beta }{2}\left( \frac{X_{i}-\mu }{{\sigma }}\right) ^{2}\right\} -%
\frac{1}{\left( 1+\beta \right) ^{3/2}}\right) =0 ,\\ 
\frac{\partial }{\partial \sigma }\frac{1}{\sigma ^{\beta }\left( 2\pi
\right) ^{\frac{\beta }{2}}}\left( \frac{1}{n\beta }\sum_{i=1}^{n}\exp
\left\{ -\frac{\beta }{2}\left( \frac{X_{i}-\mu }{{\sigma }}\right)
^{2}\right\} -\frac{1}{\left( 1+\beta \right) ^{3/2}}\right) =0.
\end{array}%
\right.
\end{equation*}%
Simple calculations yield the expressions%
\begin{equation*}
\boldsymbol{\boldsymbol{J}_\beta(\btheta_0)}=\frac{1}{%
\sqrt{1+\beta }\left( 2\pi \right) ^{\beta /2}\sigma ^{2+\beta }}\left( 
\begin{array}{cc}
\frac{1}{1+\beta } & 0 \\ 
0 & \frac{\beta ^{2}+2}{\left( 1+\beta \right) ^{2}}%
\end{array}%
\right),
\end{equation*}%
and 
\begin{equation*}
\boldsymbol{K}_\beta(\btheta)=\frac{1}{\sigma ^{2+2\beta
}\left( 2\pi \right) ^{\beta }}\left( \frac{1}{(1+2\beta )^{3/2}}\left( 
\begin{array}{cc}
1 & 0 \\ 
0 & \frac{4\beta ^{2}+2}{1+2\beta }%
\end{array}%
\right) -\left( 
\begin{array}{cc}
0 & 0 \\ 
0 & \frac{\beta ^{2}}{(1+\beta )^{3}}%
\end{array}%
\right) \right) .
\end{equation*}%
Therefore,%
\begin{equation*}
\boldsymbol{M}^{T}(\btheta_0)\boldsymbol{J}_\beta^{-1}(%
\btheta_0)\boldsymbol{K}_\beta(\btheta_0)%
\boldsymbol{J}_\beta^{-1}(\btheta_0)\boldsymbol{M}(%
\btheta_0)=\sigma ^{2}\frac{\left( \beta +1\right) ^{3}}{%
\left( 2\beta +1\right) ^{\frac{3}{2}}},
\end{equation*}
and 
\begin{eqnarray*}
W_{n} &=&n m ^{T}\left( \wtheta \right)
\left( 
\boldsymbol{M}^{T}(\wtheta)\boldsymbol{J}_\beta^{-1}(%
\wtheta)\boldsymbol{K}_\beta(\wtheta)%
\boldsymbol{J}_\beta^{-1}(\wtheta)\boldsymbol{M}(%
\wtheta)
\right)
^{-1}m\left( \wtheta\right) \\
&=&n\frac{\left( \widehat{\mu }_\beta-\mu_0\right) ^{2}\left( 2\beta
+1\right) ^{\frac{3}{2}}}{\widehat\sigma_\beta^{2}\left( \beta +1\right) ^{3}},
\end{eqnarray*}%
and on the basis of Theorem \ref{theorem_composite}, we have 
\begin{equation}
W_{n}=n\frac{\left( \widehat{\mu }_\beta-\mu_0\right) ^{2}\left(
2\beta +1\right) ^{\frac{3}{2}}}{\widehat\sigma_\beta^{2}\left( \beta +1\right) ^{3}}%
\underset{n \rightarrow \infty }{\overset{\mathcal{L}}{%
\longrightarrow }}\chi_{1}^{2}. \label{norm_comp}
\end{equation}
We observe that for $\beta=0$ we get the classical Wald test statistic for testing the hypothesis mentioned in (\ref{A}).

\subsection{The Weibull Distribution}
\label{sec:weibull}
While the normal model is the most important model where the test statistic described in Section \ref{example_normal} would be useful, it is also important to explore the applicability of these tests in other models to demonstrate the general nature of the method. In testing composite hypotheses, therefore, we have included numerical results based on the Weibull distribution in our subsequent numerical study, together with the results on the normal model. Here we describe the proposed Wald-type test statistics for the Weibull case. 
The probability density function of $\mathcal{W}(\sigma,p)$, a two parameter Weibull distribution, is given by
\bee
f_{\boldsymbol{\theta }}(x) = \frac{p}{\sigma} \left(\frac{x}{\sigma}\right)^{p-1} \exp\left\{ -\left(\frac{x}{\sigma}\right)^p\right\}, \ x>0,
\eee
where $\boldsymbol{\theta }%
=(\sigma, p)^{T}$, and the parameter space is given by $\Theta =\{(\sigma, p
)|\sigma \in {\mathbb{R}}^{+}, p \in {\mathbb{R}}^{+}\}$. A comprehensive comparison of different estimation methods for the Weibull distribution is given in \cite{teimouri2013comparison}. 
We are interested in testing
\begin{equation}
H_{0}:\sigma =\sigma _{0}\text{ against }H_{1}:\sigma \neq \sigma _{0} , \label{w}
\end{equation}
where $p$ is a nuisance parameter. Let us consider the
function $m(\boldsymbol{\theta })=\sigma -\sigma _{0}$. Then, as in the normal case which was considered in Section \ref{example_normal}, the null hypothesis $H_{0}$ can be
written as
\begin{equation*}
H_{0}:m(\boldsymbol{\theta })=0,
\end{equation*}
and  $\boldsymbol{M}(\btheta) =\left( 1,0\right) ^{T}.$

Let us define
\bee
\xi_{\alpha, \beta}( \boldsymbol{\theta }) = \int_0^\infty \left(\frac{x}{\sigma}\right)^{\alpha} f_{\boldsymbol{\theta }}^\beta(x) dx ,
\eee
and
\bee
\eta_{\alpha, \beta, \gamma}( \boldsymbol{\theta }) = \int_0^\infty \left(\frac{x}{\sigma}\right)^{\alpha} \left[ \log \left(\frac{x}{\sigma}\right)\right]^{\beta} f_{\boldsymbol{\theta }}^\gamma(x) dx .
\eee
It can be shown that
\be
\xi_{\alpha, \beta}( \boldsymbol{\theta })
=\left(\frac{p}{\sigma}\right)^{\beta-1} \beta^{-\frac{\beta p - \beta + \alpha +1}{p}} 
         \Gamma\left(\frac{\beta p - \beta + \alpha +1}{p}\right) ,
\label{xi}
\ee
and
\be
\eta_{\alpha, \beta, \gamma}( \boldsymbol{\theta }) =  \sigma \left(\frac{p}{\sigma}\right)^\gamma \int_0^\infty y^{\alpha + \gamma p - \gamma} 
            ( \log y)^\beta \exp( -\gamma y^p) dy  ,            
\label{eta}            
\ee
where $\Gamma(\cdot)$ denote the gamma function.
Note that $\xi_{\alpha, \gamma}( \boldsymbol{\theta }) = \eta_{\alpha, 0, \gamma}( \boldsymbol{\theta })$. For $\beta \neq 0$ the value of $\eta_{\alpha, \beta, \gamma}( \boldsymbol{\theta })$ is calculated using numerical integration.
Let us define
\bee
\boldsymbol{R}_\gamma( \boldsymbol{\theta }) = \int_0^\infty \boldsymbol{u}_{\boldsymbol{\theta }}(x) \boldsymbol{u}_{\boldsymbol{\theta }}^T(x) f_{\boldsymbol{\theta }}^\gamma(x) dx = \left(\begin{array}{cc}
 r_{11} & r_{12} \\  r_{12} & r_{21}
\end{array}\right),
\eee
where $\boldsymbol{u}_{\boldsymbol{\theta }}(x)$, the score function of the Weibull distribution, is given by
\bee
\boldsymbol{u}_{\boldsymbol{\theta }}(x) = \frac{\partial \log f_{\boldsymbol{\theta }}(x)}{\partial \boldsymbol{\theta}}
= \left(\begin{array}{c}
  -\frac{p}{\sigma} + \frac{p}{\sigma}\left(\frac{x}{\sigma}\right)^{p} \\
  \frac{1}{p} + \log\left(\frac{x}{\sigma}\right) - \left(\frac{x}{\sigma}\right)^p \log \left(\frac{x}{\sigma}\right) 
\end{array}\right).
\eee
Then it can be shown that
\bee
r_{11} = \left(\frac{p}{\sigma}\right)^2 \left\{  \xi_{0, \gamma}( \boldsymbol{\theta })
 - 2 \xi_{p, \gamma}( \boldsymbol{\theta })  +   \xi_{2p, \gamma}( \boldsymbol{\theta }) \right\} ,
\eee
\bee
r_{12} =  \frac{p}{\sigma}   \left\{ -\frac{1}{p} \xi_{0, \gamma}( \boldsymbol{\theta }) - \eta_{0, 1, \gamma}( \boldsymbol{\theta })  + 2 \eta_{p, 1, \gamma}( \boldsymbol{\theta })
+ \frac{1}{p}  \xi_{p, \gamma}( \boldsymbol{\theta })  -   \eta_{2p, 1, \gamma}( \boldsymbol{\theta })      \right\} ,
\eee
and
\baa
r_{22} &=& \frac{1}{p^2} \xi_{0, \gamma}( \boldsymbol{\theta }) + \eta_{0, 2, \gamma}( \boldsymbol{\theta })  +   \eta_{2p, 2, \gamma}( \boldsymbol{\theta })   + \frac{2}{p} \eta_{0, 1, \gamma}( \boldsymbol{\theta })   - 2 \eta_{p, 2, \gamma}( \boldsymbol{\theta })  - \frac{2}{p} \eta_{p, 1, \gamma}( \boldsymbol{\theta }) .
\eaa
Now 
\be
\boldsymbol{J}_\gamma( \boldsymbol{\theta }) = \int_0^\infty \boldsymbol{u}_{\boldsymbol{\theta }}(x) \boldsymbol{u}_{\boldsymbol{\theta }}^T(x) f_{\boldsymbol{\theta }}^{1+\gamma}(x) dx = \boldsymbol{R}_{1+\gamma}( \boldsymbol{\theta }),
\label{J}
\ee
\be
\boldsymbol{K}_\gamma( \boldsymbol{\theta }) = \int_0^\infty \boldsymbol{u}_{\boldsymbol{\theta }}(x) \boldsymbol{u}_{\boldsymbol{\theta }}^T(x) f_{\boldsymbol{\theta }}^{1+2\gamma}(x) dx = \boldsymbol{R}_{1+2\gamma}( \boldsymbol{\theta }).
\label{K}
\ee
Then some little algebra shows that the proposed Wald-type test statistic is given by
\begin{align*}
&  W_{n}=n(\widehat{\sigma}_{\beta}-\sigma_{0})^{2}\left[   \boldsymbol{M}^T(\btheta)
J_{\beta}^{-1}(\widehat{\sigma}_{\beta},\widehat{p}_{\beta})K_{\beta
}(\widehat{\sigma}_{\beta},\widehat{p}_{\beta})J_{\beta}^{-1}(\widehat{\sigma
}_{\beta},\widehat{p}_{\beta})\boldsymbol{M}(\btheta)\right]  ^{-1}\\
&  =\frac{n(\widehat{\sigma}_{\beta}-\sigma_{0})^{2}\left[  r_{\beta
+1}^{(1,1)}(\widehat{\sigma}_{\beta},\widehat{p}_{\beta})r_{\beta+1}%
^{(2,2)}(\widehat{\sigma}_{\beta},\widehat{p}_{\beta})-\left(  r_{\beta
+1}^{(1,2)}(\widehat{\sigma}_{\beta},\widehat{p}_{\beta})\right)  ^{2}\right]
^{2}}{%
\begin{pmatrix}
-r_{\beta+1}^{(2,2)}(\widehat{\sigma}_{\beta},\widehat{p}_{\beta}) &
r_{\beta+1}^{(1,2)}(\widehat{\sigma}_{\beta},\widehat{p}_{\beta})
\end{pmatrix}%
\begin{pmatrix}
r_{2\beta+1}^{(1,1)}(\widehat{\sigma}_{\beta},\widehat{p}_{\beta}) &
r_{2\beta+1}^{(1,2)}(\widehat{\sigma}_{\beta},\widehat{p}_{\beta})\\
r_{2\beta+1}^{(1,2)}(\widehat{\sigma}_{\beta},\widehat{p}_{\beta}) &
r_{2\beta+1}^{(2,2)}(\widehat{\sigma}_{\beta},\widehat{p}_{\beta})
\end{pmatrix}%
\begin{pmatrix}
-r_{\beta+1}^{(2,2)}(\widehat{\sigma}_{\beta},\widehat{p}_{\beta})\\
r_{\beta+1}^{(1,2)}(\widehat{\sigma}_{\beta},\widehat{p}_{\beta})
\end{pmatrix}
}\\
&  =\frac{n\left\{  (\widehat{\sigma}_{\beta}-\sigma_{0})\left(
\frac{\widehat{p}_{\beta}}{\widehat{\sigma}_{\beta}}\right)  ^{\beta}\left[
\widetilde{r}_{\beta+1}^{(1,1)}(\widehat{p}_{\beta})\widetilde{r}_{\beta
+1}^{(2,2)}(\widehat{p}_{\beta})-\left(  \widetilde{r}_{\beta+1}%
^{(1,2)}(\widehat{p}_{\beta})\right)  ^{2}\right]  \right\}  ^{2}}{%
\begin{pmatrix}
-\widetilde{r}_{\beta+1}^{(2,2)}(\widehat{p}_{\beta}) & \widetilde{r}%
_{\beta+1}^{(1,2)}(\widehat{p}_{\beta})
\end{pmatrix}%
\begin{pmatrix}
\widetilde{r}_{2\beta+1}^{(1,1)}(\widehat{p}_{\beta}) & \widetilde{r}%
_{2\beta+1}^{(1,2)}(\widehat{p}_{\beta})\\
\widetilde{r}_{2\beta+1}^{(1,2)}(\widehat{p}_{\beta}) & \widetilde{r}%
_{2\beta+1}^{(2,2)}(\widehat{p}_{\beta})
\end{pmatrix}%
\begin{pmatrix}
-\widetilde{r}_{\beta+1}^{(2,2)}(\widehat{p}_{\beta})\\
\widetilde{r}_{\beta+1}^{(1,2)}(\widehat{p}_{\beta})
\end{pmatrix}
},
\end{align*}%
where
\begin{align*}
\widetilde{r}_{\gamma}^{(1,1)}(\widehat{p}_{\beta}) &  =\varepsilon_{0,\gamma
}(\widehat{p}_{\beta})-2\varepsilon_{\widehat{p}_{\beta},\gamma}%
(\widehat{p}_{\beta})+\varepsilon_{2\widehat{p}_{\beta},\gamma}(\widehat{p}%
_{\beta}),\\
\widetilde{r}_{\gamma}^{(1,2)}(\widehat{p}_{\beta}) &  =-\frac{1}%
{\widehat{p}_{\beta}}\varepsilon_{0,\gamma}(\widehat{p}_{\beta})+\left(
\log\widehat{p}_{\beta}+\frac{1}{\widehat{p}_{\beta}}\right)  \varepsilon
_{\widehat{p}_{\beta},\gamma}(\widehat{p}_{\beta})-\log\widehat{p}_{\beta
}\varepsilon_{2\widehat{p}_{\beta},\gamma}(\widehat{p}_{\beta})-\kappa
_{0,1,\gamma}(\widehat{p}_{\beta})+\kappa_{\widehat{p}_{\beta},1,\gamma
}(\widehat{p}_{\beta}),\\
\widetilde{r}_{\gamma}^{(2,2)}(\widehat{p}_{\beta}) &  =\frac{1}%
{\widehat{p}_{\beta}}\varepsilon_{0,\gamma}(\widehat{p}_{\beta})-\frac
{2}{\widehat{p}_{\beta}}\log\widehat{p}_{\beta}\varepsilon_{\widehat{p}%
_{\beta},\gamma}(\widehat{p}_{\beta})+(\log\widehat{p}_{\beta})^{2}%
\varepsilon_{2\widehat{p}_{\beta},\gamma}(\widehat{p}_{\beta})\\
&  +\frac{2}{\widehat{p}_{\beta}}\kappa_{0,1,\gamma}(\widehat{p}_{\beta
})+\kappa_{0,2,\gamma}(\widehat{p}_{\beta})-2\log\widehat{p}_{\beta}%
\kappa_{\widehat{p}_{\beta},1,\gamma}(\widehat{p}_{\beta}),\\
\varepsilon_{\alpha,\gamma}(p)  & =\gamma^{\frac{(p-1)\gamma+\alpha+1}{p}%
}\Gamma\left(  \frac{(p-1)\gamma+\alpha+1}{p}\right)  ,\\
\kappa_{\alpha,\delta,\gamma}(\widehat{p}_{\beta})  & =\widehat{p}_{\beta}%
\int_{0}^{\infty}(\log y)^{\delta}y^{(\widehat{p}_{\beta}-1)\gamma+\alpha}%
\exp\{-\gamma y^{\widehat{p}_{\beta}}\}dy.
\end{align*}

\section{Simulation Study}\label{sec:simulation}

In this section we present a simulation study to analyze the behavior of
the proposed Wald-type test statistics introduced in this paper with some classical
procedures for the same problems. We pay special attention to the robustness
issue.

\subsection{The Case of the Simple Null Hypothesis}
We have proposed  Wald-type test statistics for the simple null hypothesis
in (\ref{3.2}). We will now study the performance of the test  through simulation 
in case of the exponential model. The expression for the test statistic is simplified
in (\ref{3.7}). We want to test the hypothesis $H_0:\theta =2$
against the alternative $H_{1}:\theta \neq 2$. In the first case we have 
generated data from the $\mathcal{E}xp(2)$ distribution, and  the observed level (measured as the proportion of test statistics
exceeding the chi-square critical value in 10,000 
 replications) are presented in Figure \ref{fig:exp}(a). We have taken
 three proposed Wald-type test statistics for $\beta= 0.1,\ 0.2$ and 0.5, denoted by $W(\beta)$, and
 compared with the classical Wald test statistic. It is clear that the observed levels
 of all these tests are very close to the nominal level of 0.05. Note that 
 the classical Wald test statistic is a special case of the proposed family of test 
 statistics corresponding to $\beta=0$.
 
 Next, the same hypotheses were tested when the data were generated from the $%
\mathcal{E}xp(1)$ distribution. The observed power (obtained in a similar manner as
above) of the test is presented in Figure \ref{fig:exp}(b) under the alternative  hypothesis. 
It is noticed that the classical Wald statistic performs best in terms of the power of the test, and for large $\beta$ the test
statistics show slightly poor performance. However, this discrepancy decreases
rapidly with the sample size, and by the time $n\geq 50$, all the observed
powers are practically equal to one. In any case for $\beta=0.1$ or 0.2 the tests are almost as powerful as the classical Wald test.

\begin{figure}[tbp]
\centering%
\begin{tabular}{rl}
\includegraphics[height=7.5cm, width=8cm]{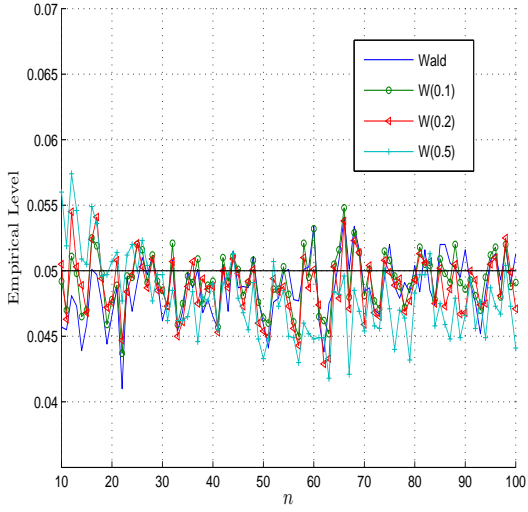}\negthinspace &
\negthinspace \includegraphics[height=7.5cm, width=8cm]{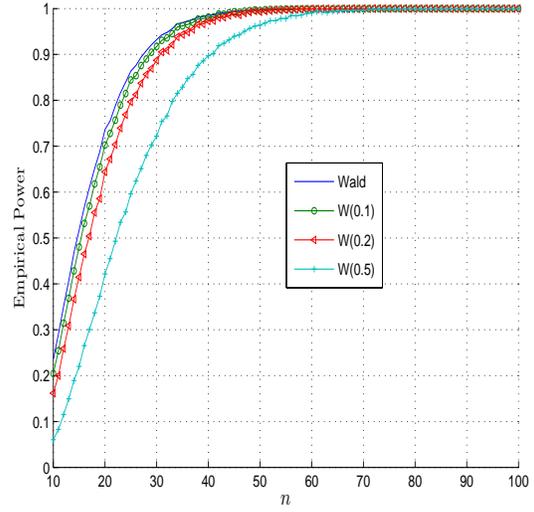} \\
\multicolumn{1}{c}{\textbf{(a)}} & \multicolumn{1}{c}{\textbf{(b)}} \\
\includegraphics[height=7.5cm, width=8cm]{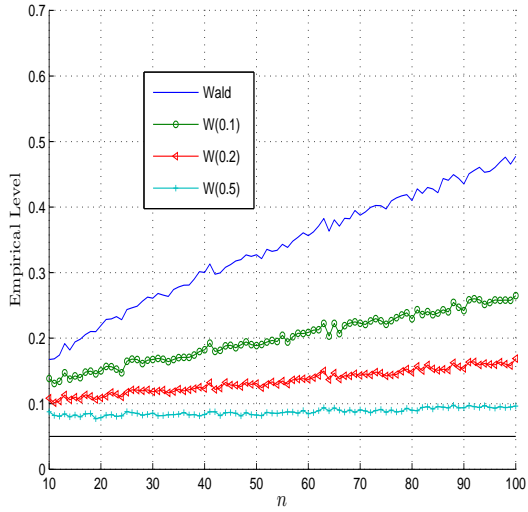}\negthinspace &
\negthinspace \includegraphics[height=7.5cm, width=8cm]{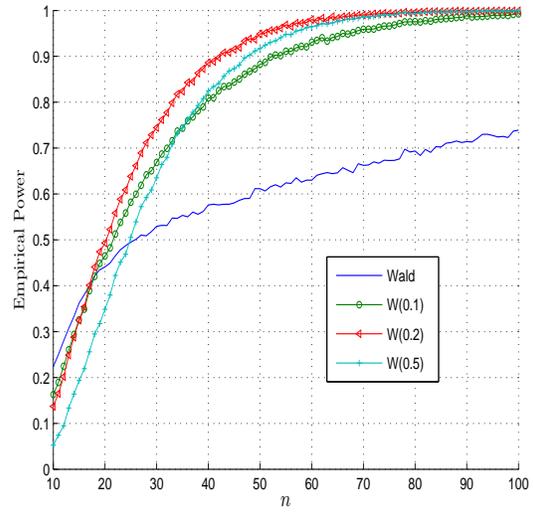} \\
\multicolumn{1}{c}{\textbf{(c)}} & \multicolumn{1}{c}{\textbf{(d)}}%
\end{tabular}%
\caption{
(a) Simulated levels of different tests for pure exponential data; (b) simulated powers of different tests for pure exponential data;
(c) simulated levels of different tests for contaminated exponential data; (d) simulated powers of different tests for contaminated exponential data.
}
\label{fig:exp}
\end{figure}

To evaluate the stability of the level and the power of the tests under
contamination, we repeated the simulations for $H_0:\theta =2$ against $H_0:\theta
\neq 2$ with data generated from the exponential mixture $0.95\mathcal{E}xp(2)+0.05\mathcal{E}xp%
(10)$, and also tested the same hypotheses with data
generated from $0.95\mathcal{E}xp(1)+0.05\mathcal{E}xp(10)$ distribution.
 The results are given in Figures \ref{fig:exp}(c) and \ref{fig:exp}(d) respectively.
 
In this case there is a huge inflation in the observed level
of the classical Wald test statistic and somewhat smaller inflation in the tests with small values of $\beta$. But as $\beta$ increases
the resistant nature of the tests are clearly apparent. For larger values of $\beta$ the levels turn out to be more acceptable. The
opposite behavior is seen in case of power. There appears to be a complete
breakdown in power for the classical Wald test, but the power remains much more
stable for the proposed Wald-type test statistics for moderate and large values of $\beta$.

On the whole it appears to be fair to claim that for sample sizes equal to
or larger than 40 the efficiency of many of our proposed Wald-type tests are very close to
the efficiency of the classical Wald test, at least in the example studied here, but the robustness properties of our tests
are often significantly better than the classical Wald test in terms of maintaining the
stability of both the level and power.

\subsection{The Case of the Composite Null Hypothesis}
\subsubsection{The Normal Case} \label{sec:normal}
To explore the performance of our proposed Wald-type test statistics in case of
the composite null hypothesis, we have performed a simulation study for the family of
test statistics given in Example \ref{example_normal}. We considered the hypothesis $H_0:\mu =0$
against the alternative $H_{1}:\mu \neq 0$ with $\sigma ^{2}$ unknown when
data were generated from the $\mathcal{N}(0,1)$ distribution. Subsequently,
the same hypotheses were tested when the data were generated from the $%
\mathcal{N}(-1,1)$ distribution. The results are
given in Figures \ref{fig:normal}(a) and \ref{fig:normal}(b). In either case the
nominal level was $0.1$, and all tests were replicated 10,000  times.

It may be noticed that all the tests with large values of $\beta$ are slightly liberal for very small sample sizes and lead to
somewhat inflated observed levels. However, the observed levels settle down reasonably with
increasing sample size. The observed powers of the tests as given in
Figure \ref{fig:normal}(b) are, in fact, extremely close. In very small sample
sizes our proposed test statistics have slightly higher power than the Wald test, but this
must be a consequence of the observed levels of these tests being higher
than the latter for such sample sizes. On the whole the proposed Wald-type tests appear to
be quite competitive to the classical Wald test for pure normal data.

\begin{figure}[tbp]
\centering%
\begin{tabular}{rl}
\includegraphics[height=7.5cm, width=8cm]{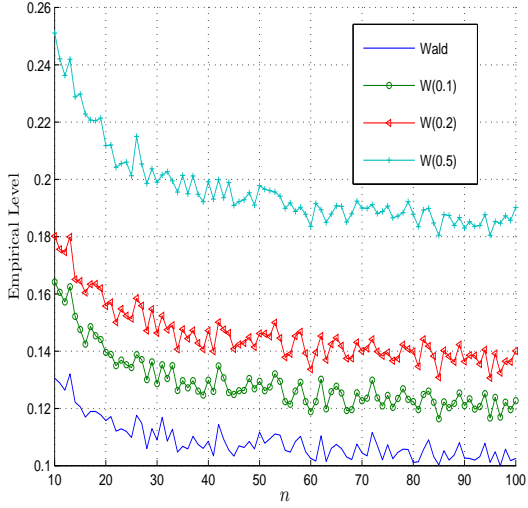}\negthinspace &
\negthinspace \includegraphics[height=7.5cm, width=8cm]{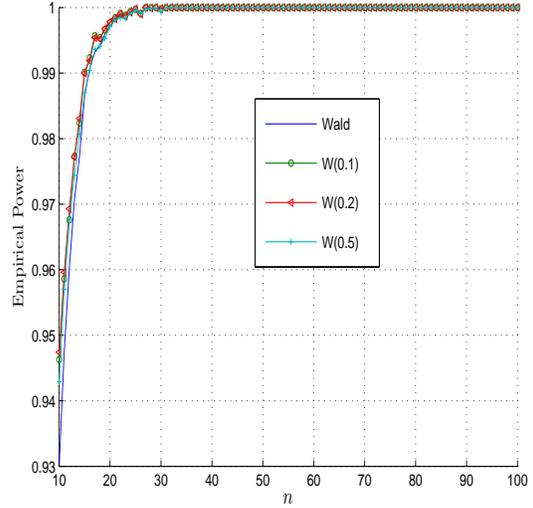} \\
\multicolumn{1}{c}{\textbf{(a)}} & \multicolumn{1}{c}{\textbf{(b)}} \\
\includegraphics[height=7.5cm, width=8cm]{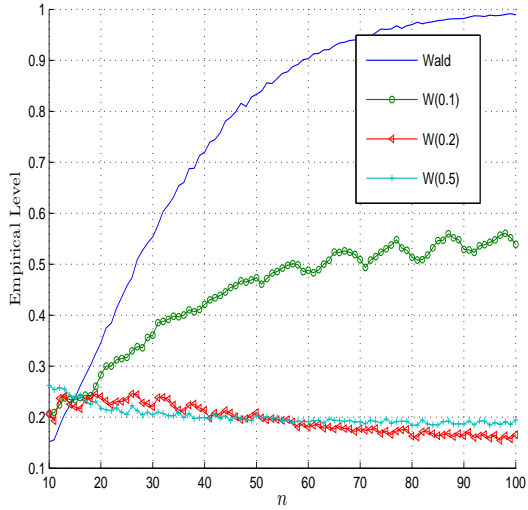}\negthinspace &
\negthinspace \includegraphics[height=7.5cm, width=8cm]{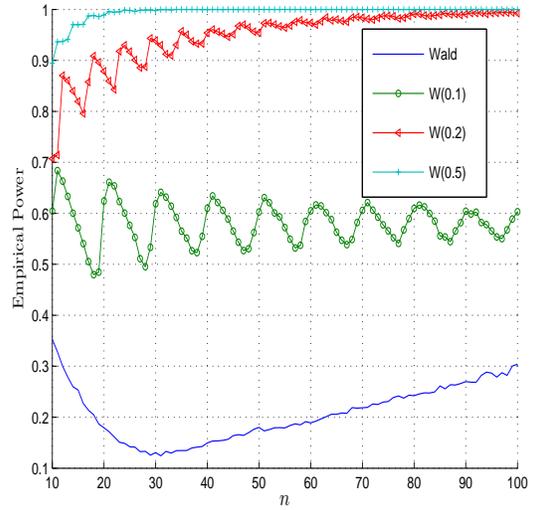} \\
\multicolumn{1}{c}{\textbf{(c)}} & \multicolumn{1}{c}{\textbf{(d)}}%
\end{tabular}%
\caption{
(a) Simulated levels of different tests for pure normal data; (b) simulated powers of different tests for pure normal data;
(c) simulated levels of different tests for contaminated normal data; (d) simulated powers of different tests for contaminated normal data.
}
\label{fig:normal}
\end{figure}

Now, we show the performance of the proposed tests under
contaminated data. So, we have tested the same hypothesis, but the data have been
generated from the $0.9\mathcal{N}(0,1)+0.1\mathcal{N}%
(10,1)$ distribution. The observed levels are shown in Figure \ref{fig:normal}(c).
 We notice that there is a drastic and severe inflation in the observed level
of the classical Wald test. As $\beta$ increases, however,
the resistant nature of the tests are clearly apparent. By the time $\beta
= 0.2$, the levels have already settled down to acceptable values. 

Finally, we have generated data from the normal mixture $0.9\mathcal{N}(-1,1)+0.1\mathcal{N}(10,1)$,
and the power functions are plotted in Figure \ref{fig:normal}(d). There is a complete
breakdown in power for the classical Wald test and the proposed Wald-type tests corresponding to the small values of $\beta$, but the power remains quite
stable for values of $\beta$ equal to 0.2 or greater.

A similar conclusion as the previous simulation study can be drawn in this case.
For sample sizes equal to
or larger than 30 the efficiency of the proposed Wald-type tests for $\beta$ greater than 0.2 are very close to
the efficiency of the classical Wald test, but the robustness properties of those tests
are much better than the Wald test in terms  of both the level and power.

In this simulation we have taken the 10\% significance level instead of the usual 5\%. Moreover, the contamination proportion is also different from the previous section. This is just to demonstrate that our tests produce analogous results in a variety of simulation set ups. The findings for the 5\% simulation results in this case, not presented here, are along similar lines. 

\subsubsection{The Weibull Case}
As we have mentioned before, it is important to demonstrate the properties of the proposed method in models other than the normal so that one has a better idea about the scope of the method. Accordingly we performed tests of composite hypotheses under the Weibull model in the spirit of Section \ref{sec:normal}. 
Let us consider the hypothesis defined in (\ref{w}), where $\sigma_0$ is taken to be 1.5. In the first study we have 
generated data from the $\mathcal{W}(1.5,1.5)$ distribution. The plot for the observed level for the hypothesis $H_0:\sigma = 1.5$ against the two sided alternative is given in Figure \ref{fig:weibull}(a), where we have used 10,000 replications. Next
the same hypotheses were tested when the data were generated from the $%
\mathcal{W}(1,1.5)$ distribution. The observed power function is plotted in Figure \ref{fig:weibull}(b) for different values of $\beta$. The powers are remarkably close. In all cases the nominal level was $0.05$.

To evaluate the stability of the level and the power of the tests under
contamination, we repeated the tests with data generated from the Weibull mixture $0.95\mathcal{W}(1.5,1.5)+0.05\mathcal{W}%
(10,1.5)$, and then from $0.95\mathcal{W}(1,1.5)+0.05\mathcal{W}(10,1.5)$. In either case the first, larger component is our target.
 In Figure \ref{fig:weibull}(c), the levels of the statistics under the contamination of first type are presented indicating the stability of levels for moderately large values of $\beta$. Figure \ref{fig:weibull}(d) demonstrates the stability of powers under contaminated data of the second type for the same values of $\beta$.

 \begin{figure}[tbp]
\centering%
\begin{tabular}{rl}
\includegraphics[height=7.5cm, width=8cm]{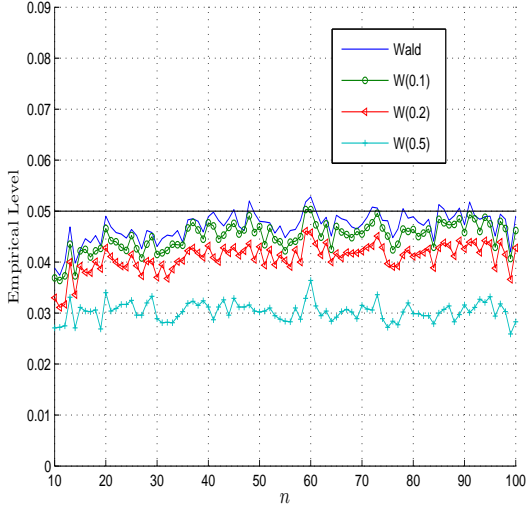}\negthinspace &
\negthinspace \includegraphics[height=7.5cm, width=8cm]{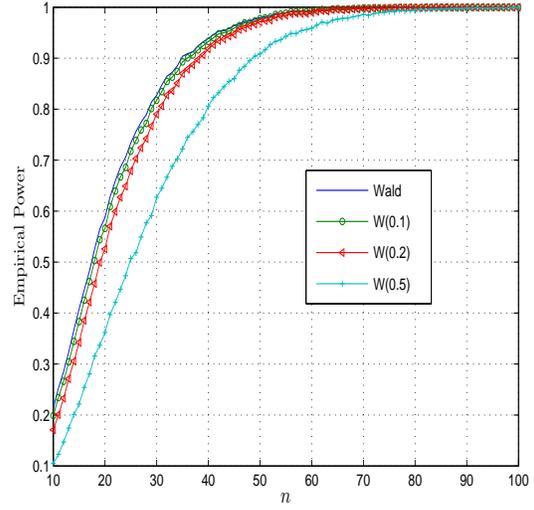} \\
\multicolumn{1}{c}{\textbf{(a)}} & \multicolumn{1}{c}{\textbf{(b)}} \\
\includegraphics[height=7.5cm, width=8cm]{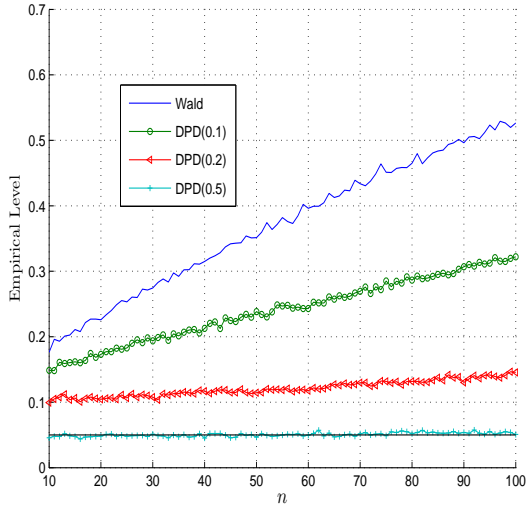}\negthinspace &
\negthinspace \includegraphics[height=7.5cm, width=8cm]{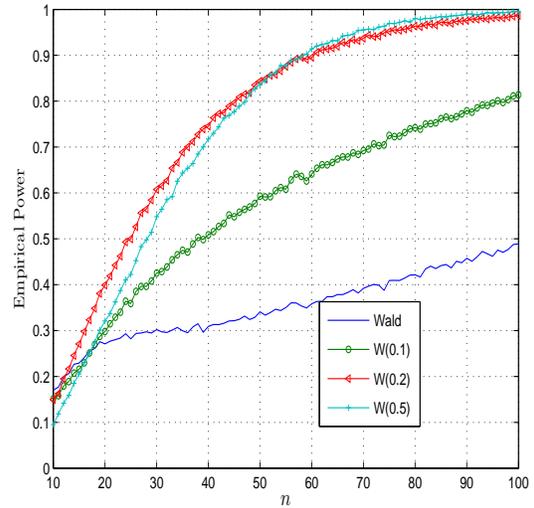} \\
\multicolumn{1}{c}{\textbf{(c)}} & \multicolumn{1}{c}{\textbf{(d)}}%
\end{tabular}%
\caption{
(a) Simulated levels of different tests for pure Weibull data; (b) simulated powers of different tests for pure Weibull data;
(c) simulated levels of different tests for contaminated Weibull data; (d) simulated powers of different tests for contaminated Weibull data.
}
\label{fig:weibull}
\end{figure}

\subsubsection{Comparison with Other Robust Tests}
In this section we compare the proposed Wald-type test with some other popular robust tests. For comparison we have used a parametric test -- the Winsorized $t$-test proposed by \cite{dixon1968approximate}, as well as three non-parametric tests -- the one-sample Kolmogorov-Smirnov test (KS-test), the two-sided Wilcoxon signed rank test and the two-sided sign test. The set up, the parameters taken for the simulation and the level of contamination are exactly the same as in Section \ref{sec:normal}. For the Winsorized $t$-test we have winsorized 15\% extreme observations from each side of the data set. It should be mentioned here that the null hypothesis is slightly different for the non-parametric tests. For the KS-test we first standardized the data using robust statistics. Then we test whether the corresponding distribution is a standard normal. The data are standardized using the transformation $Z = (X - \mu_0)/MadN$, where $\mu_0$ is the value of $\mu$ under $H_0$, and $MadN = 1.4826\ \times$ median absolute deviation about the median. In case of  Wilcoxon test and the sign test we test the null hypothesis that the median is zero, without making any assumption on the model distribution.  For comparison we have used only one proposed Wald-type test in this case, that corresponding to tuning parameter 0.2. To emphasize the robustness properties of these tests we have also included the classical Wald test in this investigation. The results are presented in Figure \ref{fig:normal_comp}. 

Figure \ref{fig:normal_comp}(a) shows that the observed levels of the Winsorized $t$-test, the KS-test and  Wilcoxon test are very close to the nominal level of 0.1 for the pure normal data. For small sample sizes the proposed Wald-type test is liberal, whereas the sign test is a little bit conservative. The observed powers of all the tests in Figure \ref{fig:normal_comp}(b) rapidly converge to unity in very small sample sizes. The result  in Figure \ref{fig:normal_comp}(c) is very interesting; it shows, for the contaminated data, all the tests except the proposed Wald-type test for $\beta  = 0.2$ fail to maintain the nominal level of 0.1. The observed level of the sign test is close to the nominal level for small sample sizes, but eventually as sample size increases it also breaks down. The powers of the tests for the contaminated data are plotted in  Figure \ref{fig:normal_comp}(d); all the robust tests show stable power. For the small sample sizes the proposed Wald-type test shows better power, however, it may be due to its inflated level in the pure data. Therefore, for the contaminated data, the proposed Wald-type test does significantly better than the others in holding on to its nominal level as well as power. On the whole, the proposed Wald-type tests are not only superior to the classical Wald test under contamination, but they also appear to be competitive or better than the other popular robust tests as far as this simulation study is concerned.

\begin{figure}[tbp]
\centering%
\begin{tabular}{rl}
\includegraphics[height=7.5cm, width=8cm]{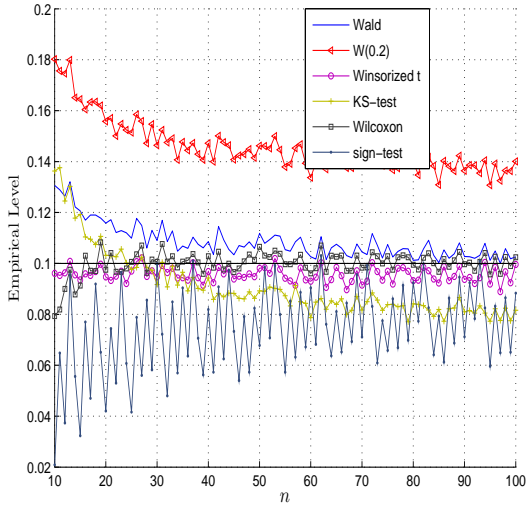}\negthinspace &
\negthinspace \includegraphics[height=7.5cm, width=8cm]{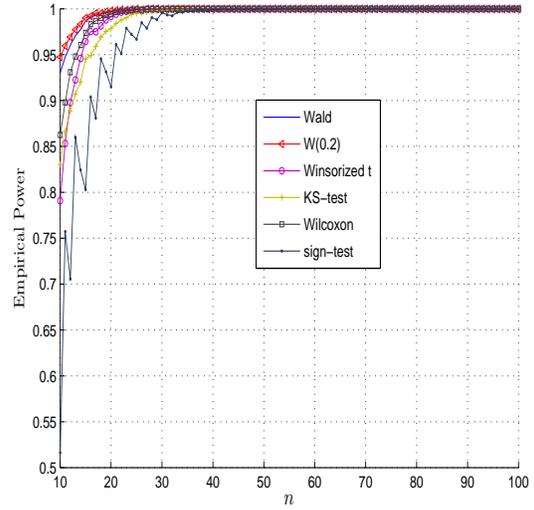} \\
\multicolumn{1}{c}{\textbf{(a)}} & \multicolumn{1}{c}{\textbf{(b)}} \\
\includegraphics[height=7.5cm, width=8cm]{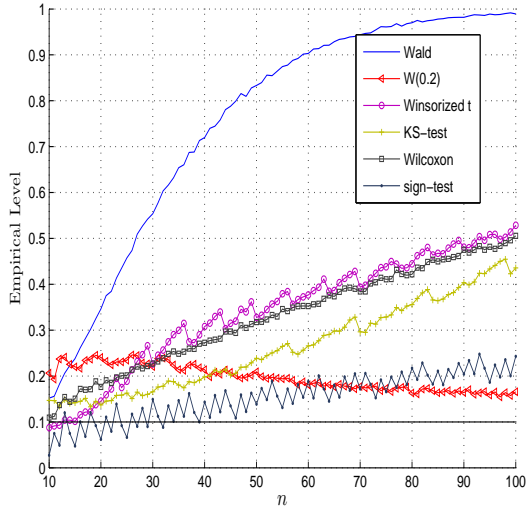}\negthinspace &
\negthinspace \includegraphics[height=7.5cm, width=8cm]{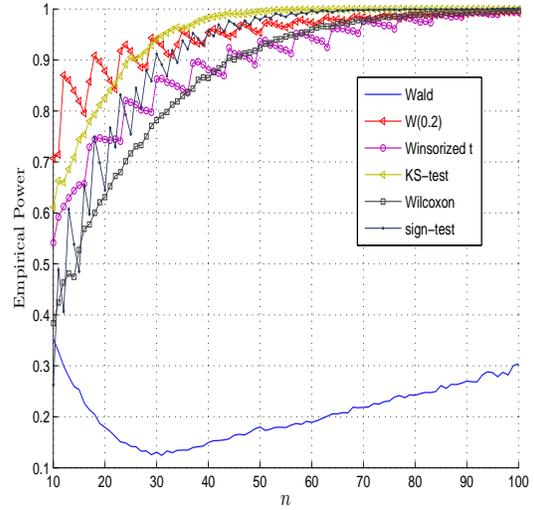} \\
\multicolumn{1}{c}{\textbf{(c)}} & \multicolumn{1}{c}{\textbf{(d)}}%
\end{tabular}%
\caption{Simulated levels and powers of different robust tests for pure and contaminated
data in case of the normal distribution.}
\label{fig:normal_comp}
\end{figure}

\subsection{Real Data Examples}

\subsubsection{Leukemia Data}
Let us consider the leukemia data set given in \cite{MR922042}, and \cite{gross1975survival}.
Table \ref{TAB:Leukemia} gives (100 times) the white blood cell counts of patients who
had acute myelogenous leukemia. We assume an exponential distribution with parameter
$\theta$  models this data. One fact that is immediately noticeable is that
there are two huge outlying observations at 1000 with respect to the
exponential model, whereas other observations appear to be reasonable
with respect to the same. The maximum likelihood estimator of $\theta$ for the full
data is 246.41, but if we delete the two outliers it comes out to be 138.75. 
We consider testing the null hypothesis $H_0: \theta = 140$
against $H_1: \theta \neq 140$ for the full data. The $p$-value of the classical Wald test is 0.0024, 
so in the presence of the outliers the classical Wald test rejects the null hypothesis. 
For the outlier deleted data the $p$-value of the classical Wald test becomes 0.9733.
This extreme turnaround illustrates that the classical Wald test is highly influenced by very small proportion of outlying observations. 
On the other hand, the proposed Wald-type tests for large values of $\beta$ lead to similar inference in all
situations. Figure \ref{fig:Leukemia_composite}(a) shows the $p$-values of the proposed Wald-type tests for the full data as well as for the  
outlier deleted data. This stable
behavior of the test statistics based on  MDPDEs for the full data
approximately coincides with the stability of the minimum  density power divergence
estimates of $\theta$ itself, obtained under an exponential model, which
is presented in Figure \ref{fig:Leukemia_composite}(b). The minimum density power
divergence estimators for the full data and the outlier deleted data are
practically identical for $\beta > 0.45$. So the
robustness of the test statistic is clearly linked to the robustness of the
estimator.

\begin{table}
\caption{Leukemia Data}
\label{TAB:Leukemia}
\begin{center}
\begin{tabular}{lcccccccccccccccc}
\hline
Patient & 1 & 2 & 3 & 4 & 5 & 6 & 7 & 8 & 9 & 10 & 11 & 12 & 13 & 14 & 15 & 16\\
Count & 23 & 7.5  & 43  & 26 &  60 &  105 &  100 &  170 &  54 &  70 &  94 &  320 &  350 &  1000 & 
520 &  1000
\\ \hline
\end{tabular}%
\end{center}
\end{table}

\begin{figure}
\centering%
\begin{tabular}{rl}
\includegraphics[height=7.5cm, width=8cm]{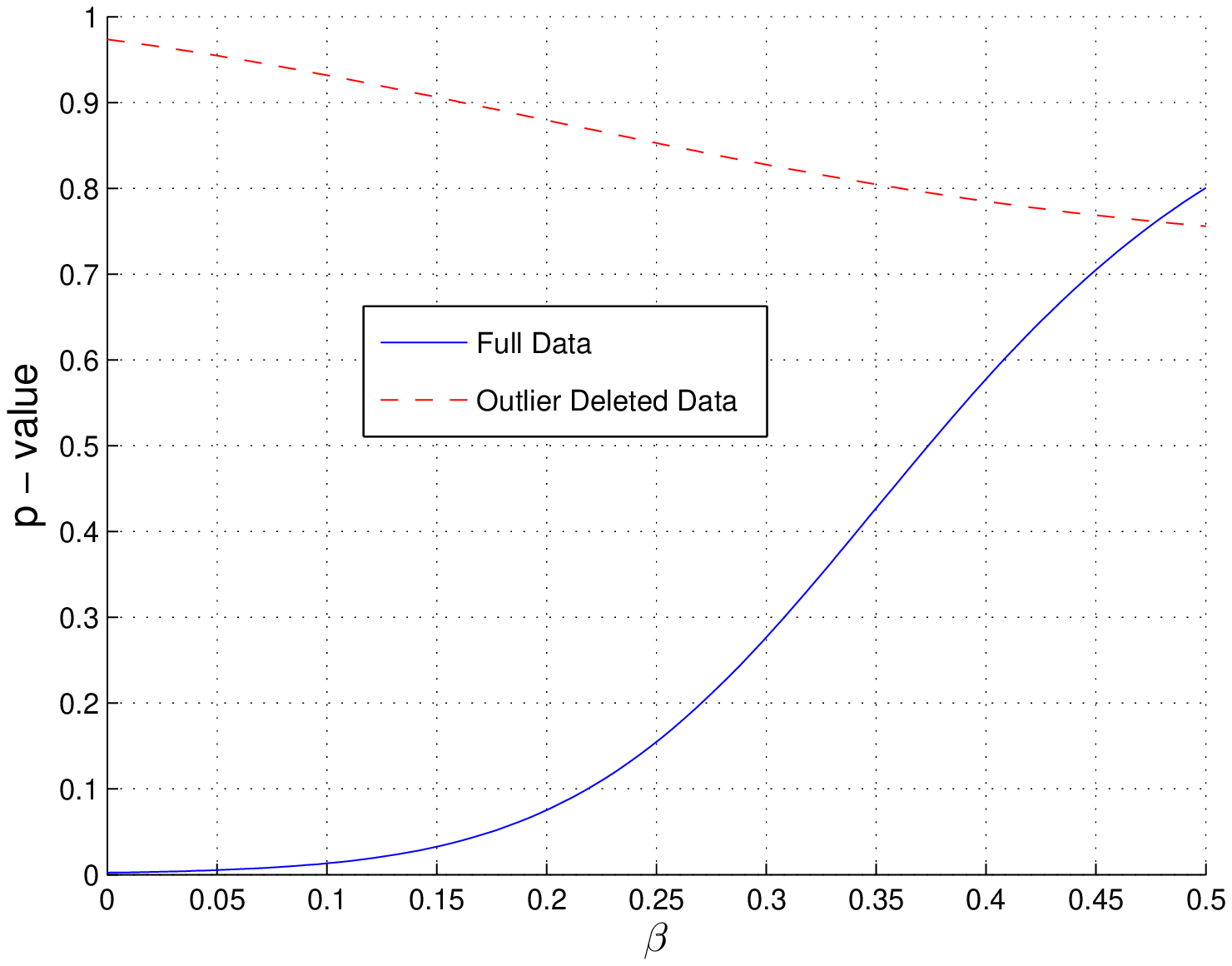}\negthinspace &
\negthinspace \includegraphics[height=7.5cm, width=8cm]{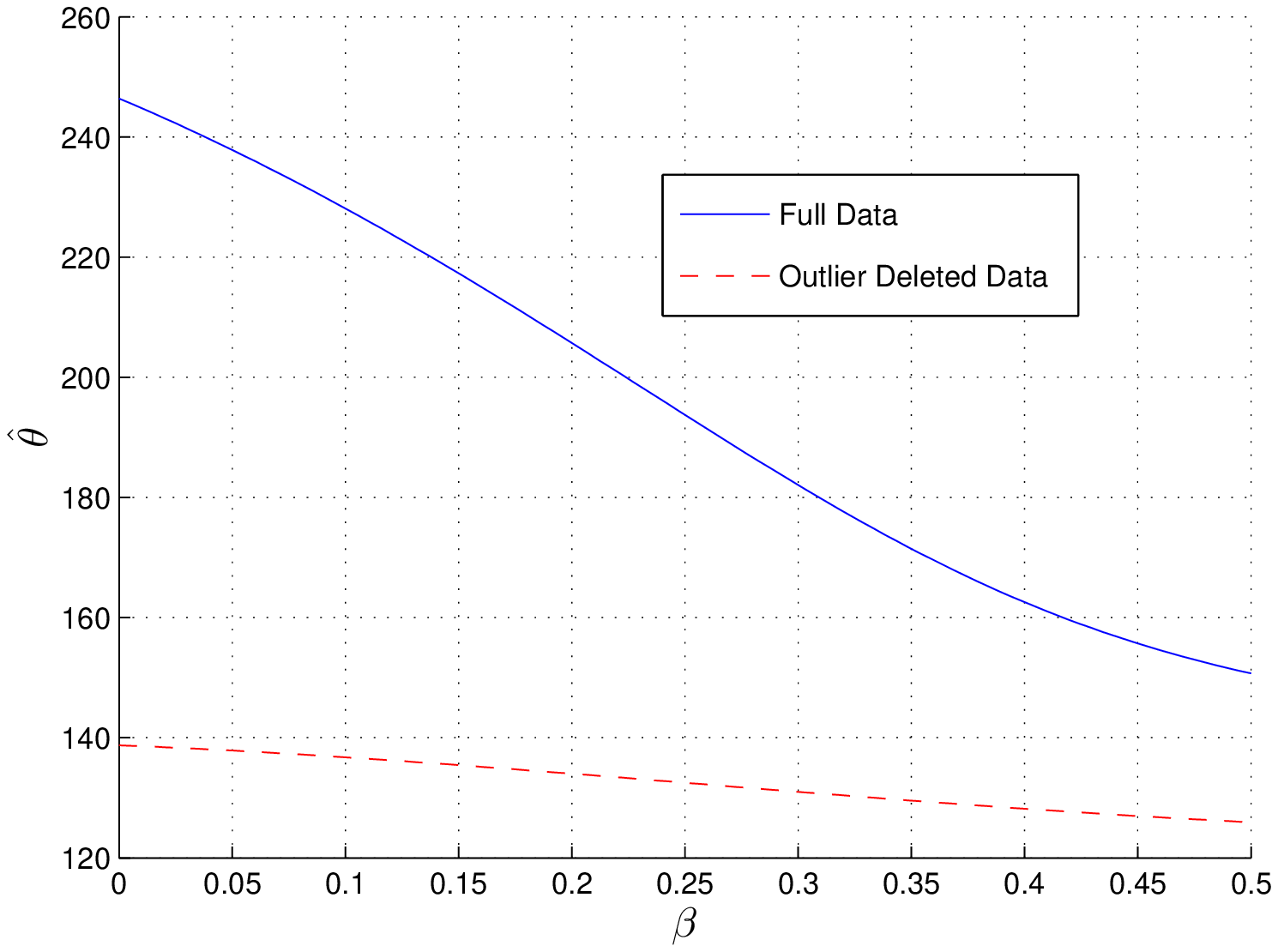} 
\end{tabular}%
\caption{(a) Two sided $p$-values of the proposed Wald-type tests, and 
(b) estimates of $\protect\theta$ for different values of $\protect\beta$ in case of the Leukemia data.}
\label{fig:Leukemia_composite}
\end{figure}

\subsubsection{Telephone-Fault Data}
\label{SEC:Telephone_example}

\cite{MR909365} and \cite{MR999667} have presented and analyzed data from an experiment to test a method of reducing faults on telephone lines. The data  consist of the ordered
differences between the inverse test rates and the inverse control rates in
14 matched pairs of areas (see Table \ref{TAB:telephone-line-faults}). For simplicity 
these data are modeled as a random sample from a normal distribution with mean $\mu$ and
standard deviation $\sigma$.  Figure \ref{normal_paper_eps} presents a
kernel density estimate for these data, the normal model fits using the estimates of $\mu$ and $\sigma$ based on the
maximum likelihood estimators, as well as the minimum density power divergence estimators (with tuning parameter $0.15$), and  the minimum
Hellinger distance estimators \citep{MR999667}. The first observation of this dataset  produces a small bump to the extreme left of the figure and is obviously a huge outlier with respect to the
normal model; the remaining 13 observations appear to be reasonable
with respect to the same. Therefore, if we ignore  the first observation,
these data would have a nice unimodal structure which could be well modeled
by an appropriate normal density. Figure \ref{normal_paper_eps} shows that the minimum Hellinger distance estimates and the minimum
density power divergence estimates automatically ignore the effect of the outlier and give  good fits for the main part of the data. On the other hand, the maximum likelihood estimates try to
be inclusive and generate a result which neither models the outlier deleted
data, nor provides a fit to the outlier component. 

\begin{figure}
\centering
\raisebox{-0cm}{\parbox[b]{8.5339cm}
{\begin{center}
\includegraphics[
height=7.5278cm,
width=8.5339cm
]{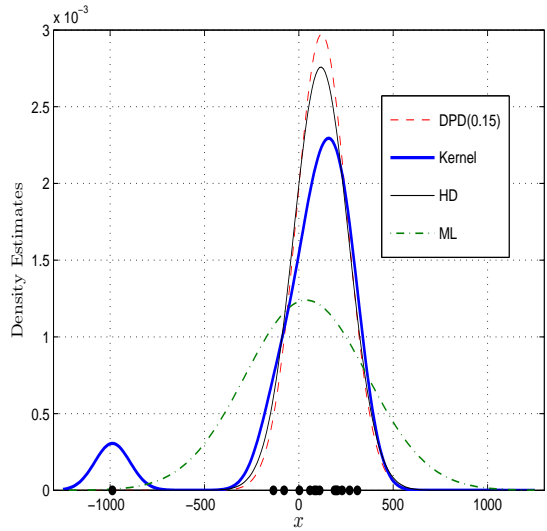}\end{center}}}
\caption{Kernel density estimate and different normal fits for the
Telephone-Fault data. }
\label{normal_paper_eps}
\end{figure}

Let us now consider testing of the null hypothesis $H_0: \mu = 0$
against $H_1: \mu \neq 0$, where $\sigma$ is unknown. 
For the full data, the classical Wald test  fails to reject the null hypothesis due to the presence of
the large outlier (two sided $p$-value is 0.6584); however the robust
Hellinger deviance test \citep{MR999667} comfortably rejects the null (two
sided $p$-value based on the chi-square null distribution is 0.0061), as
does the classical Wald test based on the cleaned data after the removal of the large
outlier (two sided $p$-value is 0.0076).

\begin{table}[h]
\caption{Telephone-Fault Data}
\label{TAB:telephone-line-faults}
\begin{center}
\begin{tabular}{lcccccccccccccc}
\hline
Pair & 1 & 2 & 3 & 4 & 5 & 6 & 7 & 8 & 9 & 10 & 11 & 12 & 13 & 14 \\
Difference & $-988$ & $-135$ & $-78$ & 3 & 59 & 83 & 93 & 110 & 189 & 197 &
204 & 229 & 289 & 310 \\ \hline
\end{tabular}%
\end{center}
\end{table}

\begin{figure}
\centering%
\begin{tabular}{rl}
\includegraphics[height=7.5cm, width=8cm]{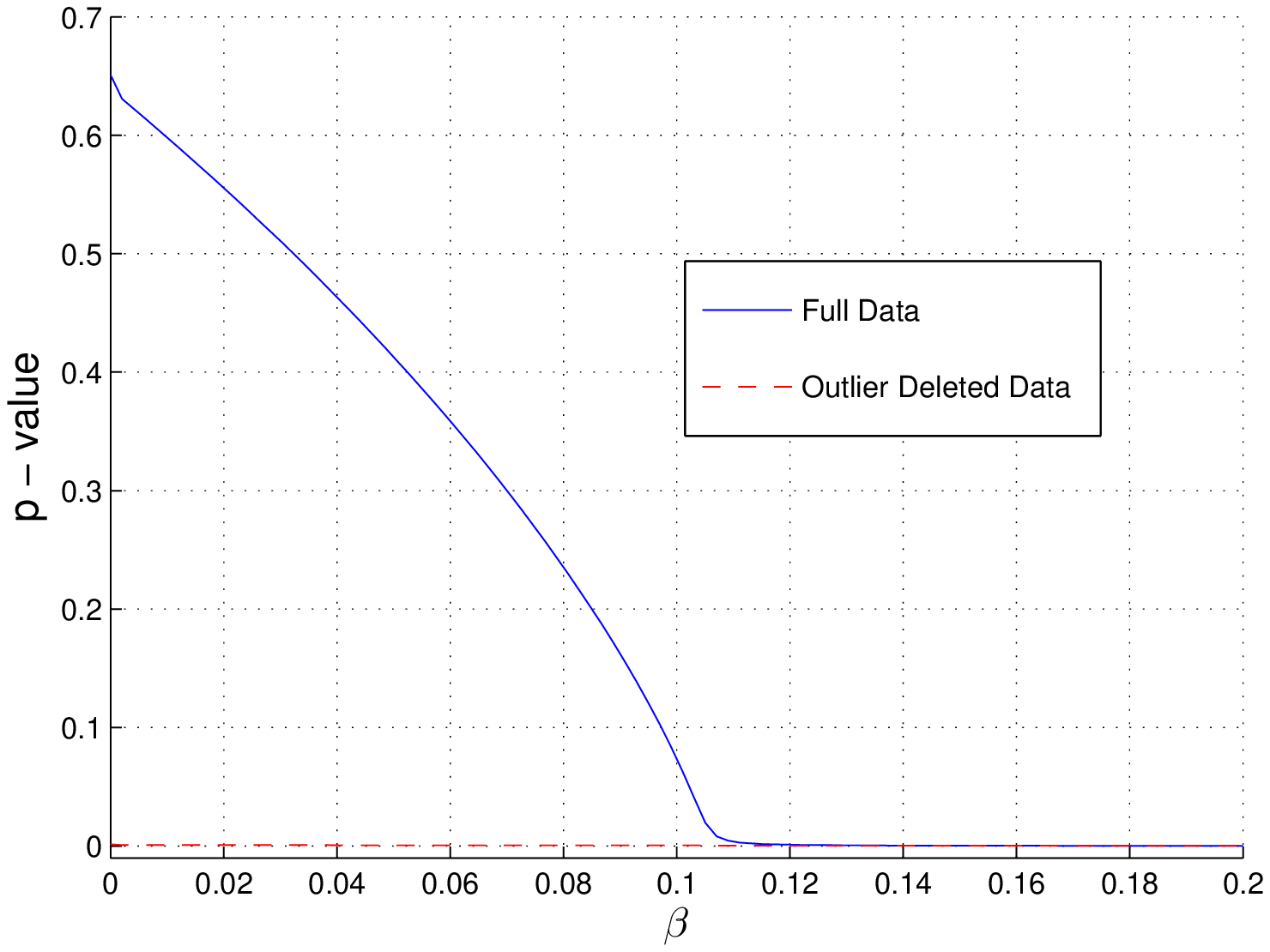}\negthinspace &
\negthinspace \includegraphics[height=7.5cm, width=8cm]{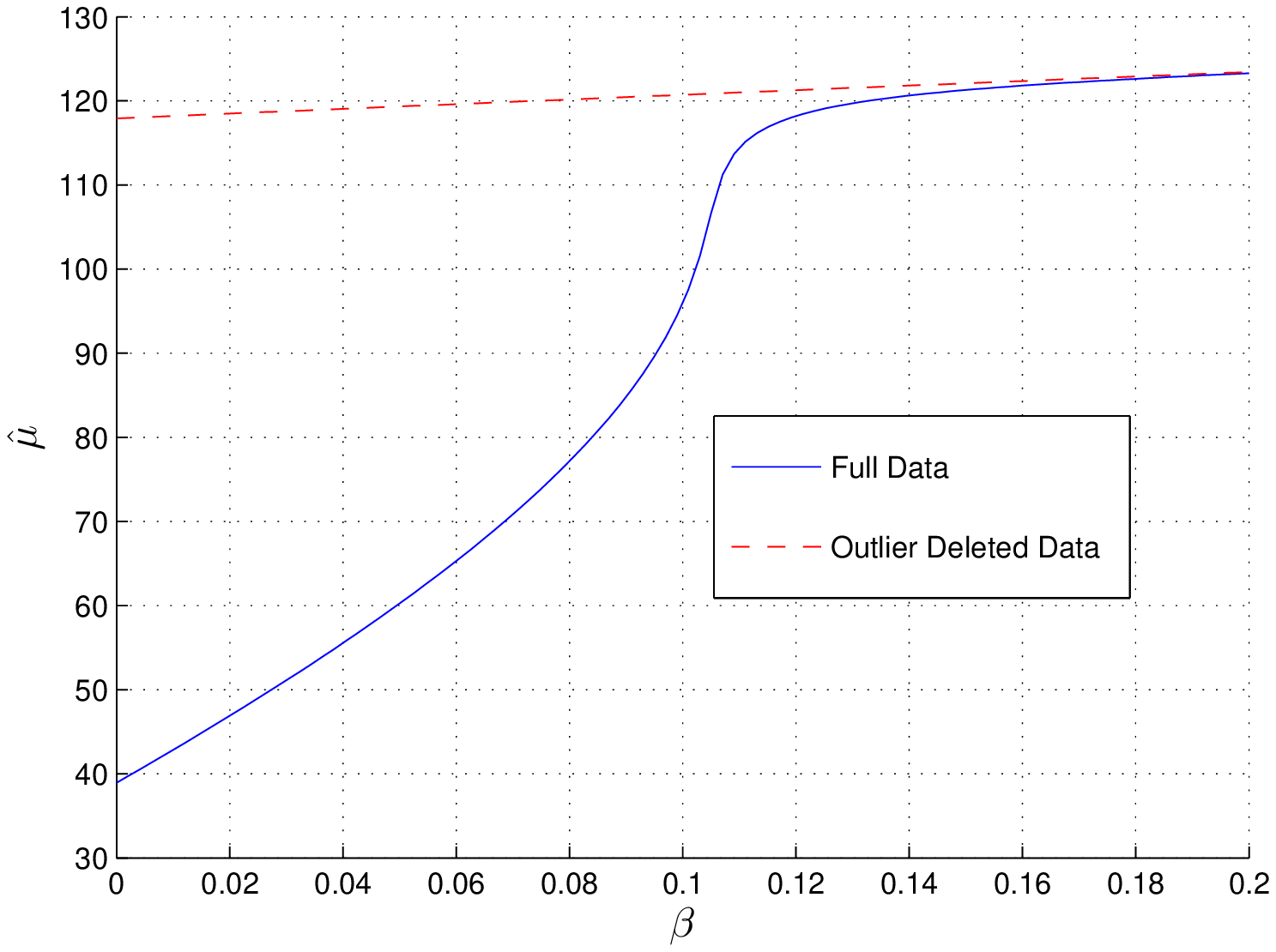} 
\end{tabular}%
\caption{(a) Two sided $p$-values of the proposed Wald-type tests, and 
(b) estimates of $\protect\mu$ for different values of $\protect\beta$ in case of the Telephone-Fault data.}
\label{fig:Tele_mu0_composite}
\end{figure}

Under the normal model, the usual estimates of $\mu$ (and $\sigma$) are
highly inflated due to the presence of the large outlier, and as a result
the likelihood ratio test under the normal model fails to reject the null
hypothesis. From the robustness perspective, this is precisely what we 
like to avoid, and here we demonstrate that proper choices of the tuning
parameter within the class of tests developed in this paper achieve this
goal. Here we analyze the performance of the proposed Wald-type tests
for different values of $\beta$. Figure \ref{fig:Tele_mu0_composite}(a) represents the $p$%
-values of the tests for different values of $\beta$ in a region of
interest. While it is clearly seen that the tests fail to reject the null
hypothesis for these data at very small values of $\beta$, the decision
turns around sharply, as $\beta$ crosses and goes beyond 0.1. On the other hand,
the $p$-values of the same test based on the outlier deleted data remain stable,
supporting rejection, at all values of $\beta$ as seen in Figure \ref{fig:Tele_mu0_composite}(a). The stable
behavior of the test statistic based on the MDPDE for the full data
approximately coincides with the stability of the MDPDE of $\mu$ itself, obtained under a two-parameter normal model, which
is presented in Figure \ref{fig:Tele_mu0_composite}(b). The minimum density power
divergence estimators for the full data and the outlier deleted data are
practically identical for $\beta > 0.12$. At least in this example, the
robustness of the test statistic is clearly linked to the robustness of the
estimator.

\begin{figure}
\centering
\raisebox{-0cm}{\parbox[b]{8.5339cm}
{\begin{center}
\includegraphics[
height=7.5278cm,
width=8.5339cm
]{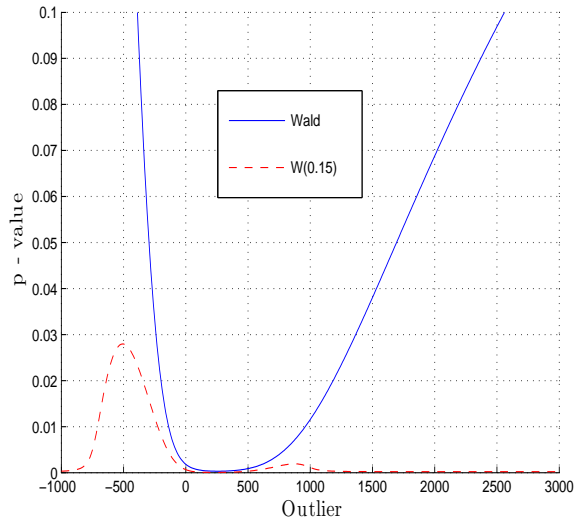}
\end{center}}}
\caption{Two sided $p$-values for the tests for the mean for the
Telephone-Fault data under the normal model against the first outlying
observation.}
\label{outliers}
\end{figure}

To further explore the robustness properties of the proposed Wald-type 
test statistics we look at the two sided $p$-values for different values of the
outlier. For this purpose we vary the first outlying observation in the
range from $-1000$ to 3000 by keeping the remaining 13 observations fixed.
Figure \ref{outliers} shows the corresponding $p$-values of the density
power divergence tests with $\beta=0.15$ as well as the
classical Wald test. It shows that initially the $p$-value of the density
power divergence test with $\beta=0.15$ increases as the first
observation moves away from the center of the data set, but after a certain
limit the test gradually nullifies the effect of the outlier. On the other
hand, the $p$-values of the classical Wald test keep on increasing with the outlier on either tail.


\subsubsection{Darwin's Plant Fertilization Data}

\label{SEC:Darwin_example}

Charles Darwin had performed an experiment which may be used to determine whether
self-fertilized plants and cross-fertilized plants have different growth
rates. In this experiment pairs of \textit{Zea mays} plants, one self and
the other cross-fertilized, were planted in pots, and after a specific time
period the height of each plant was measured. A particular sample of 15 such
pairs of plants led to the paired differences (cross-fertilized minus self
fertilized) presented in increasing order in Table \ref{TAB:Darwin-data}
(see  \citealp{darwin1878}).

\begin{table}[h]
\caption{Darwin's Plant Fertilization Data}
\label{TAB:Darwin-data}
\begin{center}
\begin{tabular}{lccccccccccccccc}
\hline
Pair & 1 & 2 & 3 & 4 & 5 & 6 & 7 & 8 & 9 & 10 & 11 & 12 & 13 & 14 & 15 \\
Difference & $-67$ & $-48$ & 6 & 8 & 14 & 16 & 23 & 24 & 28 & 29 & 41 & 49 &
56 & 60 & 75 \\ \hline
\end{tabular}%
\end{center}
\end{table}

As in the previous example, we assume a normal model for the paired
differences and test $H_0: \mu = 0$ against $H_1: \mu \neq 0$, i.e. we test whether
the mean of the paired differences is different from zero. The unconstrained
MDPDEs of $\mu$ 
under the normal model corresponding to different values of the tuning
parameter $\beta$ are presented in Figure \ref{fig:Darwin}(b).
The two negative paired differences appear to be geometrically well
separated from the rest of the data, though they are perhaps not as huge
outliers as the first observation in the Telephone-Fault data. These two
observations do have a substantial impact on the parameter estimates and the
test statistic for testing $H_0$ using the proposed Wald-type test statistic with
very small values of $\beta$, and it is instructive to compare the analysis to the
case where these two outliers have been removed from the data.
For small values of $\beta$, the two sided $p$-values of the test
statistics are drastically different for the full data and outlier deleted
cases (Figure \ref{fig:Darwin}(a)), but they get closer with increasing $%
\beta$, and they essentially coincide for $\beta \geq 0.3$. Once again
this seems to be directly linked to the robustness of the parameter
estimates; Figure \ref{fig:Darwin}(b), 
which also depicts the progression of the parameter estimates for the
outlier deleted data, clearly demonstrates that.
\begin{figure}
\centering%
\begin{tabular}{rl}
\includegraphics[height=7.5cm, width=8cm]{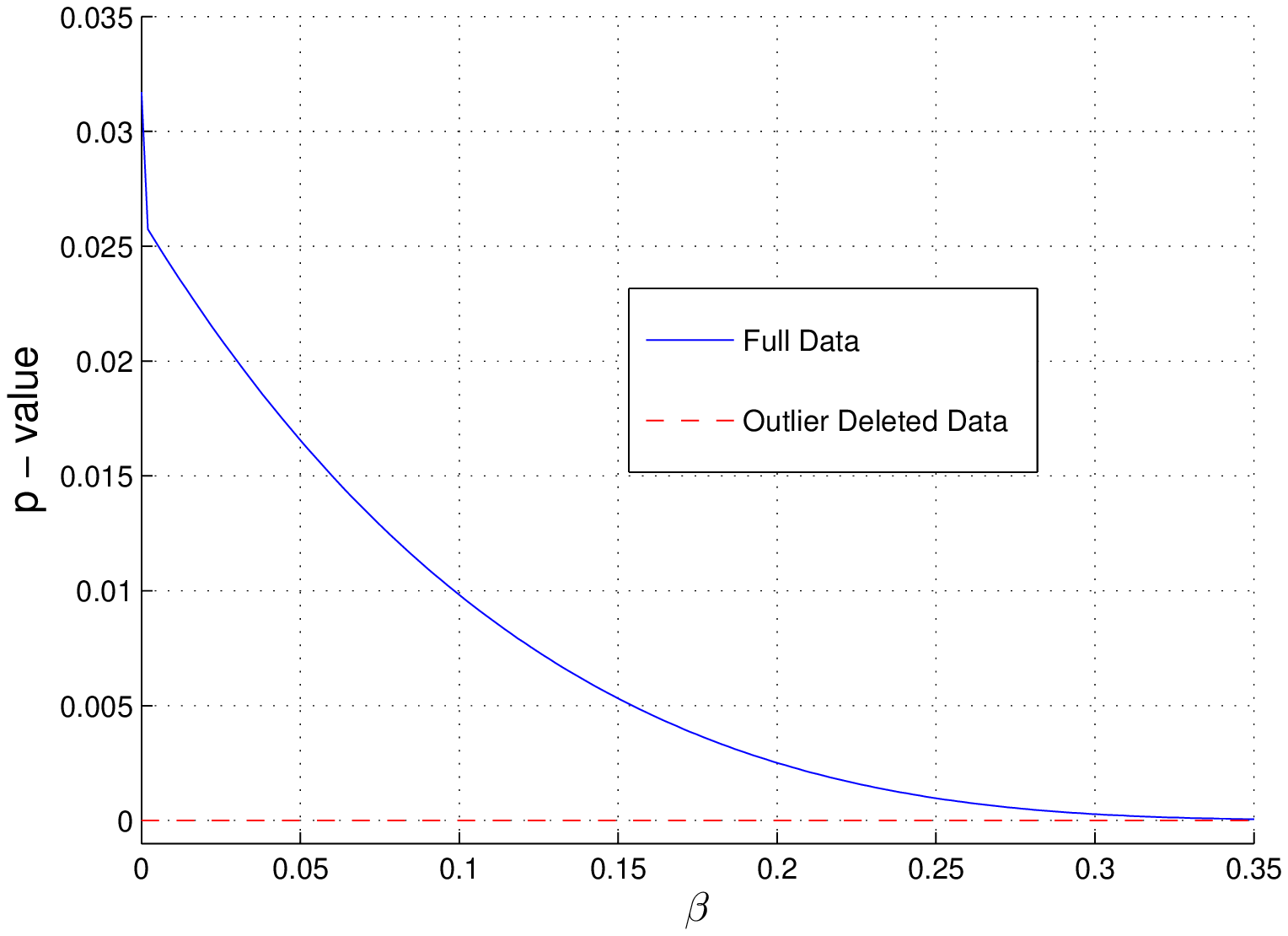}\negthinspace &
\negthinspace \includegraphics[height=7.5cm, width=8cm]{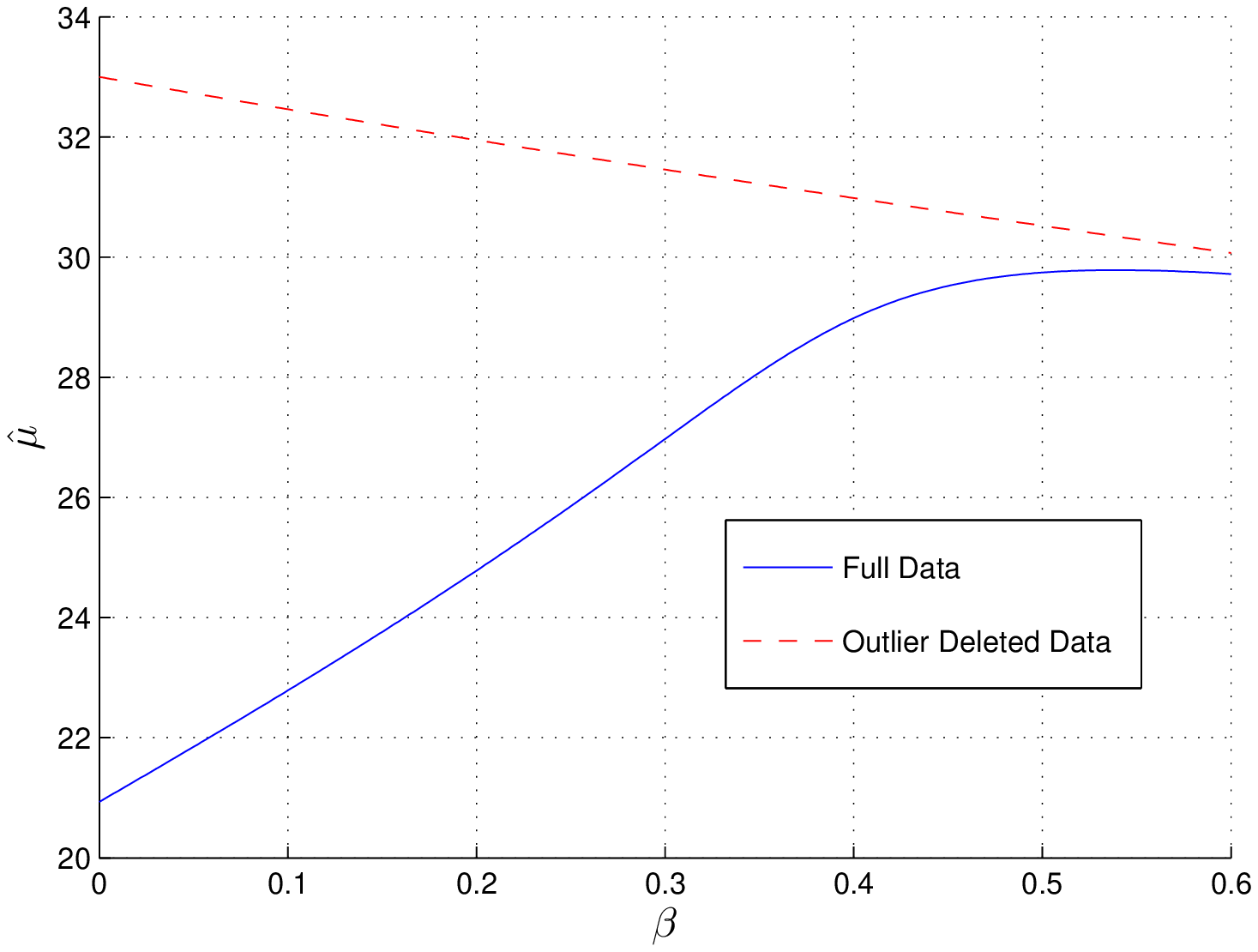} 
\end{tabular}%
\caption{(a) Two sided $p$-values of the proposed Wald-type tests, and 
(b) estimates of $\protect\mu$ for different values of $\protect\beta$ in case of Darwin's fertilization data.}
\label{fig:Darwin}
\end{figure}

\subsubsection{Air-conditioning System Failure Data}
\cite{proschan1963theoretical} gives 213 values of successive failure time of
air-conditioning system of each member of a fleet of Boeing
720 jet airplanes.  After roughly 2000 hours of service the planes received a
major overhaul; the failure interval containing major overhaul is omitted from
the listing since the length of that failure interval may have been affected by
major overhaul. The purpose of the data collection was to obtain information as to
the distribution of failure intervals for the air conditioning systems of all planes pooled together. Obvious
applications of such information would be in predicting reliability, scheduling maintenance, and providing spare parts. 
We have modeled the lifetime data as a random sample from a Weibull distribution with the shape parameter $p$ and the scale parameter $\sigma$.
However, there are few large observations in the data, which may be regarded as outliers with respect to the Weibull model. If we treat the observations which are greater than
400 (total count of such observations is only 7) as outliers we notice a clear difference in MLEs for the full data and the reduced data. The MDPDEs of $p$ for different values of $\beta$ are plotted in Figure \ref{fig:aircraft}(b). Now we expect that a similar phenomenon should be reflected in testing a hypothesis based on these estimators. We chose to test $H_0: p = 0.85$ against $H_1: p \neq 0.85$, where $\sigma$ is unknown. Figure \ref{fig:aircraft}(a) gives the $p$-values of different proposed Wald-type tests for different values of $\beta$. For the outlier deleted data all the tests clearly reject the null hypothesis, but in case of full data the classical Wald test as well as the proposed Wald-type tests with small tuning parameter $\beta$ (approximately $\beta \leq 0.2$) fail to reject the null hypothesis at 5\% level of significance. 
\begin{figure}
\centering%
\begin{tabular}{rl}
\includegraphics[height=7.5cm, width=8cm]{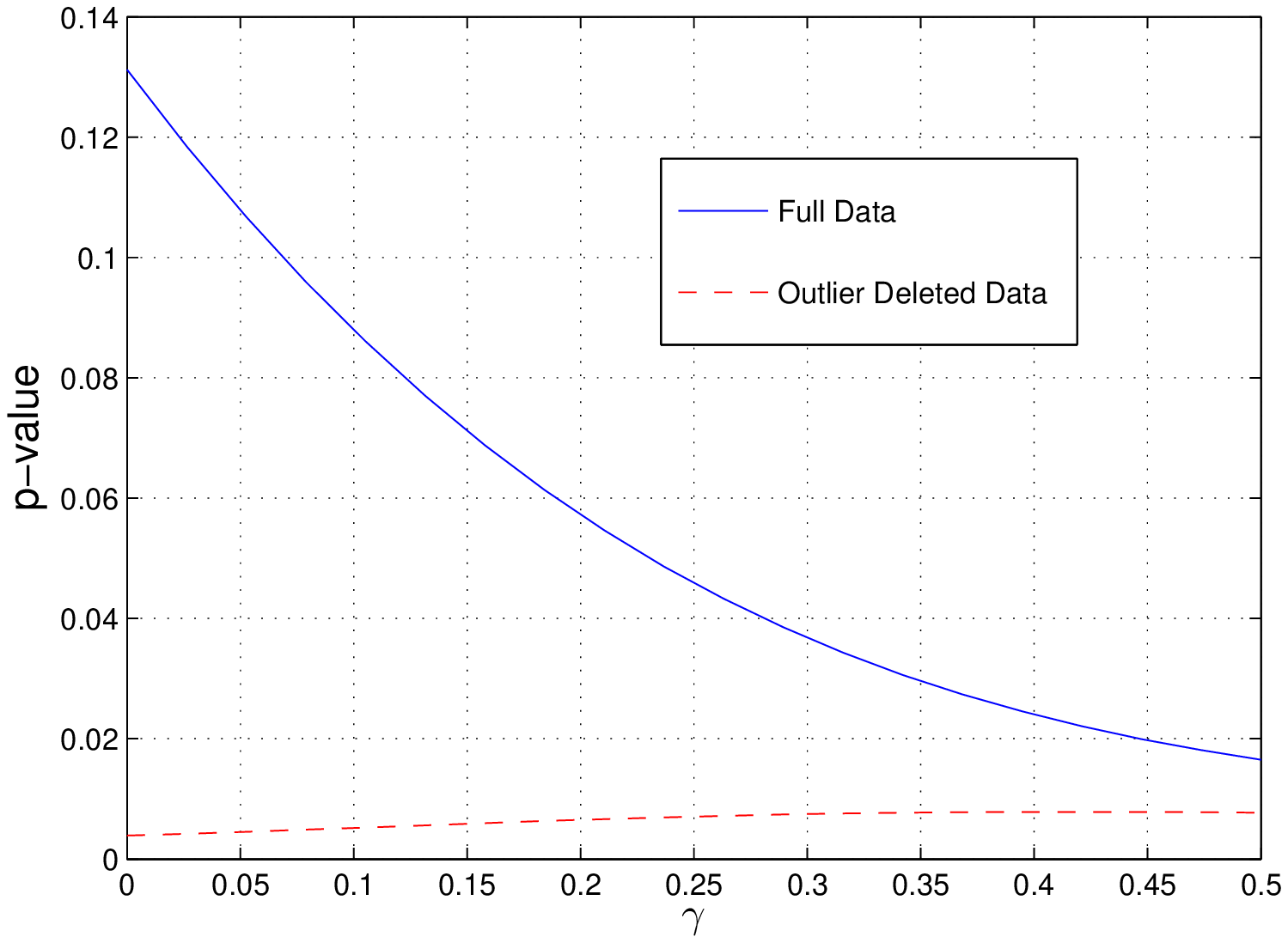}\negthinspace &
\negthinspace \includegraphics[height=7.5cm, width=8cm]{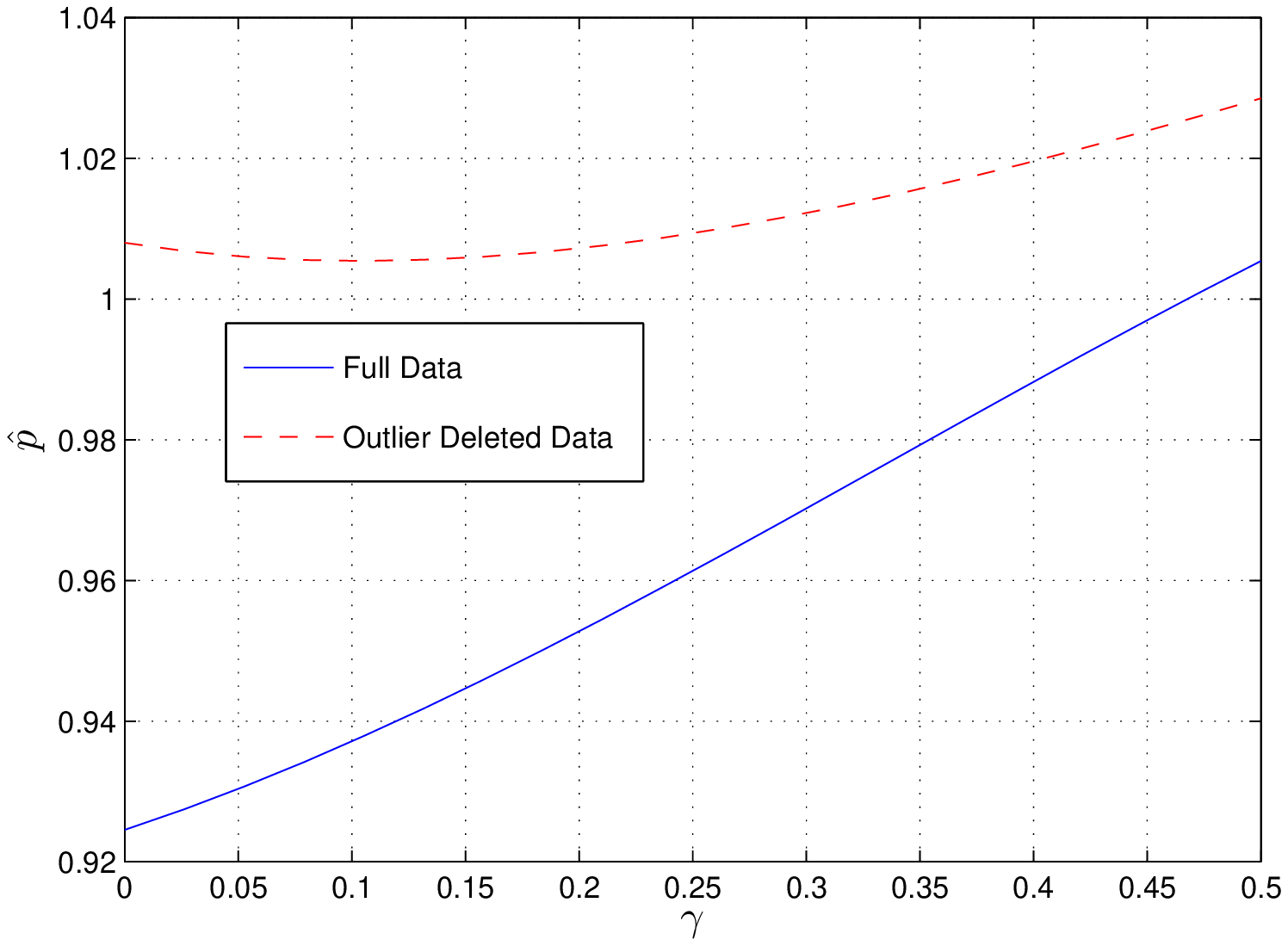} 
\end{tabular}%
\caption{(a) Two sided $p$-values of the proposed Wald-type tests, and 
(b) estimates of $\protect p$ for different values of $\protect\beta$ in case of Air-conditioning Equipment Failure data.}
\label{fig:aircraft}
\end{figure}
\subsubsection{One Sided Tests}

\label{SEC:One_Sided}

In many parametric hypothesis testing problems including testing for the mean
under the normal model the case of real interest involves a one sided
alternative. For the Telephone-Fault data problem, interest could lie in
determining whether the mean fault rate is higher than zero (rather that
simply whether it is different from zero). Darwin's interest in the
fertilization problem was, presumably, to determine whether cross
fertilization leads to a higher growth rate compared to self-fertilization;
indeed the test performed on Darwin's data by R. A. Fisher (reported in
 \citealp{fisher1935design}) considers the one sided alternative. In this subsection we
consider the relevant one sided alternatives for the two real data examples
presented earlier in this section under the normal model. For this purpose we consider the signed
Wald-type test statistic (the signed square root of the statistic presented in 
(\ref{norm_comp})) as in  \cite{MR999667}, and determine the relevant
(conservative) one sided $p$-value using the appropriate $t$-distribution.
For a scalar parameter $\theta$, given the hypothesis $H_0: \theta = \theta_0$ against $H_1: \theta > \theta_0$,
the proposed signed Wald-type statistic essentially turns out to be $W_n = \frac{n^{1/2}(\widehat\theta_\beta - \theta_0)}{\widehat\sigma_\beta}$, where
${\widehat\sigma_\beta}^2$ is the estimated variance of $n^{1/2}\widehat\theta_\beta$. The test rejects at level $\alpha$ when
$W_n$ exceeds the $100(1-\alpha)$th quantile of the $t$-distribution with $n-1$ degrees of freedom (or the corresponding
quantile of the standard normal distribution when $n$ is large). Similarly, one gets a critical region in the lower tail for the less than type alternative.

\bigskip

\noindent \textbf{Telephone-Fault data}:  \cite{MR999667} has reported the one
sided $p$-value (for the greater than alternative) in case of the Hellinger deviance test for the mean under the
normal model to be 0.0085 for the full data and 0.0093 for the outlier
deleted data. The one sided $p$-values for proposed signed Wald-type test
statistics corresponding to $\beta =0.15$ and $\beta =0.30$ are $0.0033$
and $0.0020$ respectively for the full data, and $0.0040$ and $0.0030$
respectively for the outlier deleted data, computed under the $t(13)$
distribution. On the other hand, the $p$-values for the full data and the
outlier deleted data are $0.33$ and $0.004$ for the classical one sided matched pairs 
Wald test. Clearly the outlier adversely affects the classical Wald test, but not the
proposed robust Wald-type tests. The mean of the ordered differences between the
inverse test rates and inverse control rates does appear to be actually greater than
zero.

\bigskip

\noindent \textbf{Darwin's plant fertilization data}: In this case we want
to test whether the mean of the paired differences (cross-fertilized minus
self-fertilized) is zero against the greater than alternative. The classical Wald test
performed by  \cite{fisher1935design} gives a statistic of $2.15$, with a one sided $p$%
-value of 0.025 computed under the $t(14)$ distribution; the corresponding
outlier deleted statistic has a one sided $p$-value of $0.00007$. On the
other hand, the one sided $p$-values for the proposed signed Wald-type test statistics for $%
\beta =0.15$ and $\beta =0.30$ are $0.0145$, $0.0027$ for the full data,
and $0.0001$, $0.0001$ for the outlier deleted data. Clearly the two large
outliers are influencing the decision of the classical Wald test, or the proposed Wald-type tests for very
small values of $\beta $, to the extent that the decision is reversed from
what would have been obtained with the outlier cleaned data at $1\%$ level
of significance; however larger values of $\beta $ lead to a more
consistent behavior of the tests. This data set requires stronger
downweighting compared to the Telephone-Fault data, as the outliers here are
less extreme, and therefore more difficult to identify. Under a suitable
robust test, it does appear that the mean growth of cross-fertilized plants
is higher than that of the self-fertilized plants.

\section{On the Choice of Tuning Parameter \texorpdfstring{$\beta$}{beta}}\label{sec8}

Finally, we have to give some guidance to the practitioner as to what value of the tuning
parameter $\beta$ is to be used in a particular case involving a sample of real data. Ideally
we would use the parameter $\beta = 0$ if the data were always pure, which in practice is hardly 
ever the case. Several authors (see \citealp{MR1665873} and \citealp{MR2830561}) have demonstrated
that a small positive value of $\beta$ often provides highly robust solutions with only a slight 
loss in efficiency. To analyze a particular set of real data we need a data driven choice of the 
tuning parameter which provides the necessary stability with as little loss in efficiency as possible. 

Our experience in repeated simulations and real data examples is that one generally does not require a choice of $\beta$ larger than 0.6 in most practical situations; quite often, in fact, $\beta = 0.5$ or lower is sufficient. Thus $\beta = 0.6$ can be viewed as a suitable conservative choice of the tuning parameter. However, at times this may lead
to a greater loss in efficiency than is necessary, so a more refined data driven choice of the 
tuning parameter would be helpful. In this context we prefer to begin with the MDPDE corresponding
to $\beta = 0.5$ as the pilot estimator, and determine the data driven choice of the tuning parameter
$\beta$ following the approach of \citet{Warwick}. This approach minimizes an empirical 
measure of the mean square error of the parameter estimate to determine the ``optimal'' tuning 
parameter. The pilot estimator is necessary to get the empirical measure of the bias. We note that 
\cite{MR3117102} also used the estimator corresponding to $\beta = 0.5$ as the pilot parameter
while considering robust estimation with independent but non-homogeneous data. 
One could also use $\beta = 0.6$, our conservative estimator of the tuning parameter, as the choice for constructing the pilot estimator; however our data analysis experience shows that this makes a minor difference if at all, and for reasons of efficiency we stick to the \cite{MR3117102} proposal. 

Naturally, the criterion to be considered in case of the testing problem and the estimation problem 
are not exactly the same. In case of the hypothesis testing problem, one should minimize a criterion 
which should involve a suitable linear combination of the amount of inflation in the level of the test
and  the difference of its contiguous power from one. We feel that it is much more difficult to construct an appropriate empirical 
measure of such a quantity particularly under composite alternatives. On the other hand we have demonstrated, throughout the numerical data examples, 
that the robustness of the proposed tests correspond almost exactly to the robustness of the MDPDEs. Thus
we feel that the optimal choice of $\beta$ determined as described in the previous paragraph would be a 
reasonable choice in the testing situation as well, and that remains our current recommendation. 

By way of illustration, this criterion leads to estimated optimal choices of $\beta$ to be 0.1919 and 0.5657 for 
the Telephone-Fault data and Darwin's plant fertilization data respectively.  These values are consistent with our earlier observations that for the Telephone-Fault data the first value is a massive outlier, whereas for Darwin's plant fertilization data the outliers are much closer to the overall cluster of observations, and larger values of $\beta$ will be required to eliminate their effect.

\section{Concluding Remarks}
\label{SEC:concluding}
In this paper we have constructed  generalized Wald-type tests based on  minimum density power
divergence estimators proposed by \cite{MR1665873}. 
These tests do not require any intermediate smoothing technique as in the case of
the Hellinger deviance test, so our proposed method appears to stand out
among  robust tests of hypotheses based on minimum distance methods. The proposed Wald-type tests are easy to construct, and avoid the 
problem of determining the rejection region based on  linear combinations of chi-square distributions
as in the case of \cite{MR3011625, basu2013}. These class of tests have a huge scope of application.
We have demonstrated improved results in different models including normal, exponential and Weibull
distributions. Some real data examples have also been analyzed, and we note that our methods can be useful in detecting either
kind of incorrect decision caused by a small number of extreme outliers. This is exemplified by the analysis of the Leukemia
data, where the classical test leads to a false positive, and the analysis of the Telephone-Fault data where the classical
test fails to detect a true positive. We trust that the proposed tests will be very useful practical tools for statistical data analysis.

\bigskip
\noindent
\textbf{Acknowledgments:} This work was partially supported by Grant MTM-2012-33740 and ECO-2011-25706. The authors gratefully acknowledge the suggestions of the editor, the associate editor and two anonymous referees which led to an improved version of the paper.

\bibliographystyle{abbrvnat}
 \bibliography{reference}

\section*{Appendix}

\noindent \textbf{Proof of Theorem \ref{th:main}:} 
A first order Taylor expansion of $l( \btheta) $ at $%
\wtheta$ around $\btheta^*$
gives%
\begin{equation*}
l\left( \wtheta\right) -l\left( \boldsymbol{%
\theta }^*\right) =\left( \frac{\partial l\left( \btheta%
\right) }{\partial \btheta}\right)_{\boldsymbol{\theta =\theta 
}^*}\left( \wtheta-\btheta%
^*\right) +o_p\left( \left\Vert \wtheta-%
\btheta^*\right\Vert \right).
\end{equation*}%
Now the result follows easily, since the asymptotic distribution of $n^{1/2}\left(
l\left( \wtheta\right) -l\left( \boldsymbol{%
\theta }^*\right) \right) $ coincides with the asymptotic distribution
of $n^{1/2} \left( \frac{\partial l( \btheta) 
}{\partial \btheta}\right)_{\boldsymbol{\theta =\theta }^{\ast
}}\left( \wtheta-\btheta^{\ast
}\right).$

\noindent \textbf{Proof of Theorem \ref{Th:1}:} 
We have
\begin{equation*}
n^{1/2}\left( \wtheta-\btheta%
_0\right) =n^{1/2}\left( \wtheta-\boldsymbol{%
\theta }_{n}\right) +n^{1/2}\left( \btheta_{n}-\btheta%
_0\right) =n^{1/2}\left( \wtheta-\boldsymbol{%
\theta }_{n}\right) +\boldsymbol{d}.
\end{equation*}%
Under $H_{1,n}$ it follows that
\begin{equation*}
n^{1/2}\left( \wtheta-\btheta%
_{n}\right) \underset{n\rightarrow \infty }{\overset{\mathcal{L}}{\longrightarrow }}%
\mathit{N}_p(\boldsymbol{0,J}_\beta^{-1}( \btheta_0) 
\boldsymbol{K}_\beta( \btheta_0) \boldsymbol{J}%
_\beta^{-1}( \btheta_0) ),
\end{equation*}%
and 
\begin{equation*}
n^{1/2}\left( \wtheta-\btheta%
_0\right) \underset{n\rightarrow \infty }{\overset{\mathcal{L}}{\longrightarrow }}%
\mathit{N}_p(\boldsymbol{d,J}_\beta^{-1}( \btheta_0) 
\boldsymbol{K}_\beta( \btheta_0) \boldsymbol{J}%
_\beta^{-1}( \btheta_0) ).
\end{equation*}%
On the other hand 
\begin{equation*}
W_n=\boldsymbol{X}^{T}\boldsymbol{X},
\end{equation*}%
where
\begin{equation*}
\boldsymbol{X=}\left( \boldsymbol{J}_\beta^{-1}\left( \btheta%
_0\right) \boldsymbol{K}_\beta( \btheta_0) 
\boldsymbol{J}_\beta^{-1}( \btheta_0) \right)
^{-1/2}n^{1/2}\left( \wtheta-\boldsymbol{\theta 
}_0\right),
\end{equation*}%
and under $H_{1,n}$
\begin{equation*}
\boldsymbol{X}\underset{n\rightarrow \infty }{\overset{\mathcal{L}}{\longrightarrow }}%
\mathit{N}_p\left(\left( \boldsymbol{J}_\beta^{-1}\left( \btheta%
_0\right) \boldsymbol{K}_\beta( \btheta_0) 
\boldsymbol{J}_\beta^{-1}( \btheta_0) \right)
^{-1/2}\boldsymbol{d},\boldsymbol{I}_{p\times p}\right).
\end{equation*}%
Here $\boldsymbol{I}_{p\times p}$ is the identity matrix of order $p$. Therefore 
\begin{equation*}
W_n=\boldsymbol{X}^{T}\boldsymbol{X}\underset{n\rightarrow \infty }{%
\overset{\mathcal{L}}{\longrightarrow }}\chi_{p}^{2}(\delta )
\end{equation*}%
with $\delta =\boldsymbol{d}^{T}\boldsymbol{J}_\beta^{-1}\left( \boldsymbol{%
\theta }_0\right) \boldsymbol{K}_\beta\left( \btheta%
_0\right) \boldsymbol{J}_\beta^{-1}( \btheta_0) 
\boldsymbol{d}$; here $\chi_{p}^{2}(\delta )$ denotes a non-central chi-square 
distribution with $p$ degrees of freedom and non-centrality parameter $\delta .$

\bigskip
\noindent \textbf{Proof of Theorem \ref{theorem_composite}:} 
 Let $\btheta_0 \in \Theta_0$ be the true value of $\btheta$.
Using a Taylor series expansion we get
\begin{eqnarray}
\boldsymbol{m}\left( \wtheta\right)  &=&\boldsymbol{m}%
( \btheta_0) +\boldsymbol{M}(\btheta%
_0)^{T}\left( \wtheta-\btheta%
_0\right) +o_p\left( \left\Vert \wtheta-%
\btheta_0\right\Vert \right)  \nonumber\\
&=&\boldsymbol{M}(\btheta_0)^T\left( \wtheta -\btheta_0\right) 
+o_p\left( \left\Vert \wtheta -\btheta_0\right\Vert \right) ,
\label{m}
\end{eqnarray}%
because from equation (\ref{2.5}) we have $\boldsymbol{m}( \btheta_0) \boldsymbol{=0}_{r}.$
Now, under $H_0$, 
\begin{equation*}
n^{1/2}\left( \wtheta-\btheta%
_0\right) \underset{n \rightarrow \infty }{\overset{\mathcal{L}}{%
\longrightarrow }} N_p(\boldsymbol{0}_p,\boldsymbol{J}_\beta^{-1}\left( \boldsymbol{%
\theta }_0\right) \boldsymbol{K}_\beta\left( \btheta%
_0\right) \boldsymbol{J}_\beta^{-1}( \btheta_0) ).
\end{equation*}%
Therefore, from equation (\ref{m}) we get, under $H_0$, 
\begin{equation*}
n^{1/2}\boldsymbol{m}\left( \wtheta\right) \underset%
{n \rightarrow \infty }{\overset{\mathcal{L}}{\longrightarrow }} N_r
(\boldsymbol{0}_r,\boldsymbol{M} ^{T}\left( \btheta%
_0\right) \boldsymbol{J}_\beta^{-1}(\btheta_0)%
\boldsymbol{K}_\beta(\btheta_0)\boldsymbol{J}_{\beta
}^{-1}(\btheta_0)\boldsymbol{M}\left( \boldsymbol{%
\theta }_0\right) ).
\end{equation*}%
As rank$\left( \boldsymbol{M}(\btheta) \right) = r$, we get
\begin{equation*}
n\boldsymbol{m} ^{T}\left( \wtheta\right) \left[
\boldsymbol{M}^{T}(\btheta_0)\boldsymbol{J}_\beta^{-1}(%
\btheta_0)\boldsymbol{K}_\beta(\btheta_0)%
\boldsymbol{J}_\beta^{-1}(\btheta_0)\boldsymbol{M}(%
\btheta_0)\right] ^{-1}\boldsymbol{m}\left( \widehat{\boldsymbol{%
\theta }}_\beta\right) \underset{n \rightarrow \infty }{\overset{%
\mathcal{L}}{\longrightarrow }}\chi_{r}^{2}.
\end{equation*}
Now $\boldsymbol{M}^{T}\left( \wtheta\right)\boldsymbol{J}_\beta^{-1}
\left( \wtheta\right) \boldsymbol{K}_\beta\left( \wtheta\right)%
\boldsymbol{J}_\beta^{-1}\left( \wtheta\right)\boldsymbol{M}
\left( \wtheta\right)$ is a consistent estimator of \newline
$\boldsymbol{M}^{T}(\btheta_0)\boldsymbol{J}_\beta^{-1}(%
\btheta_0)\boldsymbol{K}_\beta(\btheta_0)%
\boldsymbol{J}_\beta^{-1}(\btheta_0)\boldsymbol{M}(%
\btheta_0)$. Hence, under $H_0$,
\begin{equation*}
n\boldsymbol{m} ^{T}\left( \wtheta\right) \left[
\boldsymbol{M}^{T}\left( \wtheta\right)\boldsymbol{J}_\beta^{-1}
\left( \wtheta\right) \boldsymbol{K}_\beta\left( \wtheta\right)%
\boldsymbol{J}_\beta^{-1}\left( \wtheta\right)\boldsymbol{M}
\left( \wtheta\right)\right] ^{-1}\boldsymbol{m}\left( \widehat{\boldsymbol{%
\theta }}_\beta\right) \underset{n \rightarrow \infty }{\overset{%
\mathcal{L}}{\longrightarrow }}\chi_{r}^{2}.
\end{equation*}

\bigskip
\noindent \textbf{Proof of Theorem \ref{Th:10}:} 
We note that $l^*\left(\wtheta, \wtheta\right)$ and 
$l^*\left(\wtheta, \btheta^*\right)$ have same asymptotic distribution, because
$\widehat{\boldsymbol{%
\theta }}_\beta\overset{p}{\underset{n\rightarrow \infty }{%
\longrightarrow }}\btheta^*$. Now 
a first order Taylor expansion of $l^*\left(\wtheta, \btheta^*\right)$ at 
$\wtheta$ around $\btheta^*$
gives%
\begin{equation*}
l^*\left(\wtheta, \btheta^*\right) -
l^*\left(\btheta^*, \btheta^*\right) =
\left( \frac{\partial l^*( \btheta, \btheta^*) }{%
\partial \btheta}\right)_{\btheta =\btheta^*
}^{T} \left( \wtheta
-\btheta^*\right) +o_p\left( \left\Vert 
\wtheta -\btheta^*\right\Vert
\right) .
\end{equation*}%
So the theorem is proved from the following result
\begin{equation*}
n^{1/2}\left( \wtheta -\btheta^* \right) \underset{%
n\rightarrow \infty }{\overset{\mathcal{L}}{\longrightarrow }} N_p \left(\boldsymbol{0}_p,
\boldsymbol{J}_\beta^{-1}(\btheta^*)\boldsymbol{K}_{\beta
}(\btheta^*)\boldsymbol{J}_\beta^{-1}(\btheta^*) \right).
\end{equation*}%

\noindent \textbf{Proof of Theorem \ref{Th:11}:} 
A Taylor series expansion of $\boldsymbol{m}\left( \widehat\btheta%
_\beta\right) $ around $\btheta_{n}$ yields%
\begin{equation*}
\boldsymbol{m}\left( \wtheta\right) =\boldsymbol{m}%
\left( \btheta_{n}\right) +\boldsymbol{M} ^{T}\left( \btheta%
_{n}\right) \left( \wtheta-\btheta%
_{n}\right) +o\left( \left\Vert \wtheta-\boldsymbol{%
\theta }_{n}\right\Vert \right) .
\end{equation*}%
From (\ref{mt}) we have 
\begin{equation}
\boldsymbol{m}\left( \wtheta\right) =n^{-1/2}\boldsymbol{M} ^{T}
( \btheta_0) \boldsymbol{d} +\boldsymbol{M} ^{T}\left( \btheta%
_{n}\right) \left( \wtheta-\btheta%
_{n}\right) +o\left( \left\Vert \wtheta-\boldsymbol{%
\theta }_{n}\right\Vert \right) +o\left( \left\Vert \btheta_{n}-%
\btheta_0\right\Vert \right) . \label{the}
\end{equation}%
Under $H_{1,n}$ we get $n^{1/2}\left( o\left( \left\Vert \wtheta-%
\btheta_{n}\right\Vert \right) +o\left( \left\Vert \btheta%
_{n}-\btheta_0\right\Vert \right) \right) =o_p(1) $
and%
\begin{equation*}
n^{1/2}\left( \wtheta-\btheta%
_{n}\right) \underset{n \rightarrow \infty }{\overset{\mathcal{L}}{%
\longrightarrow }}\mathit{N}_p(\boldsymbol{0}_p,\boldsymbol{J}_\beta^{-1}\left( \boldsymbol{%
\theta }_0\right) \boldsymbol{K}_\beta\left( \btheta%
_0\right) \boldsymbol{J}_\beta^{-1}( \btheta_0) ).
\end{equation*}%
So from (\ref{the}) we have 
\begin{equation*}
n^{1/2}\boldsymbol{m}\left( \wtheta\right) \underset%
{n \rightarrow \infty }{\overset{\mathcal{L}}{\longrightarrow }}\mathit{N%
}_r(\boldsymbol{M} ^{T}( \btheta_0) \boldsymbol{d,M} ^{T}
( \btheta_0) \boldsymbol{J}_\beta^{-1}\left( 
\btheta_0\right) \boldsymbol{K}_\beta\left( \btheta%
_0\right) \boldsymbol{J}_\beta^{-1}( \btheta_0) 
\boldsymbol{M}( \btheta_0) ).
\end{equation*}%
From (\ref{eqiv}) we get, under $H_{1,n}^*$, 
\begin{equation*}
n^{1/2}\boldsymbol{m}\left( \wtheta\right) \underset%
{n \rightarrow \infty }{\overset{\mathcal{L}}{\longrightarrow }}\mathit{N%
}_r(\boldsymbol{\delta}, \boldsymbol{M} ^{T}( \btheta_0) \boldsymbol{J}%
_\beta^{-1}( \btheta_0) \boldsymbol{K}_{\beta
}( \btheta_0) \boldsymbol{J}_\beta^{-1}\left( 
\btheta_0\right) \boldsymbol{M}( \btheta_0)
).
\end{equation*}%
We apply the following result concerning quadratic forms: If $\boldsymbol{%
Z\in }\mathit{N}_k\left( \boldsymbol{\mu ,\Sigma }\right) ,$ $\boldsymbol{%
\Sigma }$ is a symmetric projection of rank $k$ and $\boldsymbol{\Sigma \mu
=\mu ,}$ then $\boldsymbol{Z}^{T}\boldsymbol{Z}$ is a chi-square
distribution with $k$ degrees of freedom and non-centrality parameter $%
\boldsymbol{\mu }^{T}\boldsymbol{\mu .}$ Here the quadratic form is 
\begin{equation*}
W_{n}=\boldsymbol{Z}^{T}\boldsymbol{Z},
\end{equation*}%
where
\begin{equation*}
\boldsymbol{Z}=n^{1/2}\boldsymbol{m}\left( \widehat\btheta_{\beta
}\right) \left[ \boldsymbol{M}^{T}(\btheta_0)\boldsymbol{J}%
_\beta^{-1}(\btheta_0)\boldsymbol{K}_\beta(\boldsymbol{%
\theta }_0)\boldsymbol{J}_\beta^{-1}(\btheta_0)%
\boldsymbol{M}(\btheta_0)\right] ^{-1/2}.
\end{equation*}%
We know
\begin{equation*}
\boldsymbol{Z}\underset{n \rightarrow \infty }{\overset{\mathcal{L}}{%
\longrightarrow }}\mathit{N}_r\left( \left[ \boldsymbol{M}^{T}(\boldsymbol{%
\theta }_0)\boldsymbol{J}_\beta^{-1}(\btheta_0)%
\boldsymbol{K}_\beta(\btheta_0)\boldsymbol{J}_{\beta
}^{-1}(\btheta_0)\boldsymbol{M}(\btheta_0)%
\right] ^{-1/2}\boldsymbol{M} ^{T}( \btheta_0) %
\boldsymbol{d,I}_r\right) ,
\end{equation*}%
where $\boldsymbol{I}_r$ is the identity matrix of order $r$. Hence the
application of the result is immediate. The non-centrality parameter is 
\begin{equation*}
\boldsymbol{d}^{T}\boldsymbol{M}( \btheta_0) \left[ 
\boldsymbol{M}^{T}(\btheta_0)\boldsymbol{J}_\beta^{-1}(%
\btheta_0)\boldsymbol{K}_\beta(\btheta_0)%
\boldsymbol{J}_\beta^{-1}(\btheta_0)\boldsymbol{M}(%
\btheta_0)\right] ^{-1}\boldsymbol{M} ^{T}\left( \btheta%
_0\right) \boldsymbol{d}.
\end{equation*}
\end{document}